\documentclass[11pt]{article}
\usepackage[top=1.05in, bottom=1.05in, left=1.0in, right=1.0in]{geometry}
\usepackage{amsfonts, amssymb, amsmath, amsthm,amsbsy,mathtools, dsfont, xcolor}
\usepackage[linktoc=page, colorlinks, linkcolor=blue, citecolor=blue]{hyperref}
\usepackage{enumerate, commath}
\usepackage{algorithm,algpseudocode}

\usepackage{authblk}
\usepackage{graphicx}

\usepackage{subcaption} 
\usepackage[font=small]{caption} 
\numberwithin{equation}{section}

\DeclareMathOperator{\E}{\mathbb{E}}

\DeclareMathOperator{\Var}{Var}

\DeclareMathOperator{\tr}{tr}

\DeclareMathOperator{\range}{range}

\DeclareMathOperator*{\argmin}{argmin}
\DeclareMathOperator{\Lap}{Lap}
\DeclareMathOperator{\ran}{ran}
\DeclareMathOperator{\Ber}{Ber}
\DeclareMathOperator{\Vol}{Vol}

\renewcommand{\Pr}[2][]{\mathbb{P}_{#1} \left\{ #2 \rule{0mm}{3mm}\right\}}

\newcommand{\ip}[2]{\langle#1,#2\rangle}

\def \N {\mathbb{N}}
\def \P {\mathbb{P}}
\def \R {\mathbb{R}}

\def \Z {\mathbb{Z}}

\def \BB {\mathcal{B}}

\def \GG {\mathcal{G}}
\def \FF {\mathcal{F}}

\def \NN {\mathcal{N}}

\def \BB {\mathcal{B}}

\def \e {\varepsilon}
\def\epsilon{\varepsilon}
\def \d {\delta}
\def \l {\lambda}
\def \s {\sigma}
\def \S {\Sigma}

\def \tran {\mathsf{T}}

\def \one {{\textbf 1}}

\def \sym {\mathrm{sym}}
\def \off {\mathrm{off}}

\newtheorem{theorem}{Theorem}[section]
\newtheorem{proposition}[theorem]{Proposition}
\newtheorem{corollary}[theorem]{Corollary}
\newtheorem{lemma}[theorem]{Lemma}
\theoremstyle{definition}

\newtheorem{definition}[theorem]{Definition}

\newtheorem{remark}[theorem]{Remark}

\begin{document}

\title{Covariance's Loss is Privacy's Gain: \\  Computationally~Efficient,~Private~and~Accurate~Synthetic~Data}

\author[1]{March Boedihardjo}
\affil[1]{Department of Mathematics, ETH Z\"urich}
\author[2]{Thomas Strohmer}
\affil[2]{Center of Data Science and Artificial Intelligence Research and Department of Mathematics, University of California Davis}  
\author[3]{Roman Vershynin}
\affil[3]{Department of Mathematics, University of California Irvine}

\maketitle


\begin{abstract}
The protection of private information is of vital importance  in data-driven research, business, and government. 
The conflict between privacy and utility has triggered intensive research in the computer science and statistics communities, who have developed a variety of methods for privacy-preserving data release. Among the main  concepts  that have emerged are anonymity and  differential privacy.  Today, another solution is gaining traction,  synthetic data. However,  the road to privacy is paved with NP-hard problems.
In this paper we focus on the NP-hard challenge  to develop a synthetic data generation method that is computationally efficient, comes with provable privacy guarantees, and rigorously quantifies data utility.
We solve a relaxed version of this problem by studying a fundamental, but a first glance completely unrelated, problem in probability concerning the concept of covariance loss. Namely, we find a nearly optimal and constructive answer to the question how much information is lost when we take conditional expectation. Surprisingly,  this excursion into theoretical probability produces mathematical techniques that allow us to derive constructive, approximately optimal solutions  to difficult applied problems concerning microaggregation, privacy, and synthetic data.

\end{abstract}

\section{Introduction}

{\em ``Sharing is caring''}, we are taught. But if we care about privacy, then we better think twice what we share.
As governments and companies are increasingly collecting vast amounts of personal information (often without the consent or knowledge of the user~\cite{zuboff2019}), it is crucial to ensure that fundamental rights to privacy of the subjects the data refer to are guaranteed\footnote{The importance of individual privacy is underscored by the fact that Article 12 of the Universal Declaration of Human Rights is concerned with privacy.}.  We are facing the problem of how to release data that are useful  to make accurate decisions and predictions without disclosing sensitive information on specific identifiable individuals.  

The conflict between privacy and utility has triggered intensive research in the computer science and statistics communities, who have developed a variety of methods for privacy-preserving data release. Among the main  concepts  that have emerged are {\em anonymity} and {\em differential privacy}~\cite{domingo2016database}.   Today, another solution is gaining traction, {\em synthetic data}~\cite{bellovin2019privacy}. However, the road to privacy is paved with NP-hard problems.
For example, finding the optimal partition into $k$-anonymous groups is NP-hard~\cite{meyerson2004complexity}. Optimal
multivariate microaggregation is NP-hard~\cite{oganian2001complexity,thaeter2020hardness} (albeit, the error metric used in these papers is different from the one used in our paper).
Moreover, assuming the existence of one-way functions, there is no polynomial time, differentially private algorithm for generating boolean synthetic data that preserves all two-dimensional marginals with accuracy $o(1)$~\cite{ullman2011pcps}.

No matter which privacy preserving strategy one pursues, in order to implement that strategy the challenge is to navigate this NP-hard privacy jungle and
develop a method that is computationally efficient, comes with provable privacy guarantees, and  rigorously quantifies data utility. This is the main  topic of our paper.



\subsection{State of the art}

{\bf Anonymity} captures the understanding that it should not be possible to re-identify any individual in the published data~\cite{domingo2016database}. One of the most popular ways in trying to ensure anonymity is via the concept of {\em $k$-anonymity}~\cite{sweeney2002k,sweeney2002achieving}. A dataset has the $k$-anonymity property if the information for each person contained in the dataset cannot be distinguished from at least $k-1$ individuals whose information also appears in the dataset. Although the privacy guarantees offered by $k$-anonymity are limited, its simplicity has made it a popular  part of the arsenal of privacy enhancing technologies, see e.g.~\cite{domingo2016database,fei2017k,liu2019voting,khan2020theta}. 
$k$-anonymity is often implemented via the concept of microaggregation~\cite{domingo2005ordinal,li2011provably,soria2014enhancing,laszlo2015iterated,domingo2016database}.
The principle of microaggregation is to partition a data set into groups of at least $k$ similar records and to replace the records in each
group by a prototypical record (e.g.\ the centroid).

Finding the optimal partition into $k$-anonymous groups is an NP-hard problem~\cite{meyerson2004complexity}. Several practical algorithms exists that  produce  acceptable empirical results, albeit without any theoretical  bounds on the information loss~\cite{domingo2005ordinal,domingo2016database,Monedero}. In light of the  popularity of $k$-anonymity, it is thus quite surprising that it is an open problem to design a computationally efficient algorithm for $k$-anonymity that comes with theoretical utility guarantees.

As $k$-anonymity is prone to various attacks,  {\em differential privacy} is generally considered a more robust type of privacy.

{\bf Differential privacy} formalizes the intuition that the presence or absence of any single individual record in the database
or data set should be unnoticeable when looking at the responses returned for the queries~\cite{dwork2014algorithmic}. 
Differential privacy is a popular and robust method that  comes with a rigorous mathematical framework and provable guarantees. 
It can protect aggregate information, but not sensitive information in general. Differential privacy is usually implemented  via  noise injection, 
where the noise level depends on the query sensitivity. However, the added noise will negatively affect utility of the released data.

As pointed out in~\cite{domingo2016database}, microaggregation is a useful primitive to find bridges between privacy models.
It is a natural idea to combine microaggregation with differential privacy~\cite{soria2014enhancing,sanchez2016utility} to address some of the privacy limitations of $k$-anonymity. As before, the fundamental question is whether there are computationally efficient methods to implement this scheme while also maintaining utility guarantees.

\smallskip

{\bf Synthetic data} are generated (typically via some randomized algorithm) from existing data such that they maintain the statistical properties of the original data set, but do so without risk of exposing sensitive information.  Combining synthetic data with differential privacy is a promising means to overcome key weaknesses of the latter~\cite{hardt2012simple,bellovin2019privacy,kearnsroth2020,liu2021leveraging}.
Clearly, we want the synthetic data to be faithful to the original data, so as to preserve utility. To quantify the faithfulness, we need some similarity metrics.  A common and natural choice for tabular data is to try to (approximately) preserve low-dimensional marginals~\cite{barak2007privacy,thaler2012faster,dworknikolov}. 

\medskip
We model the {\em true data} $x_1,\ldots,x_n$ as a sequence of $n$ points from the Boolean cube $\{0,1\}^p$,  which is a standard benchmark data model~\cite{barak2007privacy,ullman2011pcps,hardt2012simple,ping2017datasynthesizer}. 
For example, this can be $n$ health records of patients, each containing $p$ binary parameters (smoker/nonsmoker, etc.)\footnote{More generally, one can represent any categorical data
(such as gender, occupation, etc.), genomic data, or numerical data (by splitting them into intervals) on the Boolean cube via binary or one-hot encoding.}
We are seeking to transform the true data into {\em synthetic data} $y_1,\ldots,y_m \in \{0,1\}^p$ that is both differentially private and accurate. 

As mentioned before, we measure accuracy by comparing the {\em marginals} of true and synthetic data. A $d$-dimensional marginal of the true data has the form 
$$
\frac{1}{n} \sum_{i=1}^n x_i(j_1) \cdots x_i(j_d)
$$
for some given indices $j_1,\ldots,j_d \in [p]$. In other words, a low-dimensional marginal is the fraction of the patients whose $d$ given parameters all equal $1$. 
The one-dimensional marginals encode the means of the parameters, 
and the two-dimensional marginals encode the covariances. 

The {\em accuracy} of the synthetic data for a given marginal can be defined as
\begin{equation}	\label{eq: error tensor}
E(j_1,\ldots,j_d) 
\coloneqq \frac{1}{n} \sum_{i=1}^n x_i(j_1) \cdots x_i(j_d)
	- \frac{1}{m} \sum_{i=1}^m y_i(j_1) \cdots y_i(j_d).
\end{equation}
Clearly, the accuracy is bounded by $1$ in absolute value.

\subsection{Our contributions}

Our goal is to design a randomized algorithm that satisfies the following list of desiderata:
\begin{enumerate}[(i)]
\vspace*{-2mm}
\setlength{\itemsep}{-0.5ex}
\item {\bf (synthetic data):} the algorithm outputs a list of vectors $y_1,\ldots,y_m \in \{0,1\}^p$; 
\item {\bf (efficiency):} the algorithm requires only polynomial time in $n$ and $p$;
\item {\bf (privacy):} the algorithm is differentially private;
\item {\bf (accuracy):} the low-dimensional marginals of $y_1,\ldots,y_m$ are close to those of $x_1,\ldots,x_n$.
\end{enumerate}


There are known algorithms that satisfy any three of the above four requirements if we restrict the accuracy condition (iv) to two-dimensional marginals. 

Indeed, if~(i) is dropped, one can first compute the mean $\frac{1}{n}\sum_{k=1}^{n}x_{k}$ and the covariance matrix $\frac{1}{n}\sum_{k=1}^{n}x_{k}x_{k}^{T}-(\frac{1}{n}\sum_{k=1}^{n}x_{k})(\frac{1}{n}\sum_{k=1}^{n}x_{k})^{T}$, add some noise to achieve differential privacy, and output i.i.d.~samples from the Gaussian distribution with the noisy mean and covariance.

Suppose~(ii) is dropped. It suffices to construct a differentially private probability measure $\mu$ on $\{0,1\}^{p}$ so that $\int_{\{0,1\}^{p}}x\,d\mu(x)\approx\frac{1}{n}\sum_{k=1}^{n}x_{k}$ and $\int_{\{0,1\}^{p}}xx^{T}\,d\mu(x)\approx\frac{1}{n}\sum_{k=1}^{n}x_{k}x_{k}^{T}$. After $\mu$ is constructed, one can generate i.i.d.~samples $y_{1},\ldots,y_m$ from $\mu$. The measure $\mu$ can be constructed as follows: First add Laplace noises to $\frac{1}{n}\sum_{k=1}^{n}x_{k}$ and $\frac{1}{n}\sum_{k=1}^{n}x_{k}x_{k}^{T}$ (see Lemma \ref{le:Lap} below) and then set $\mu$ to be a probability measure on $\{0,1\}^{p}$ that minimizes $\|\int_{\{0,1\}^{p}}x\,d\mu(x)-(\frac{1}{n}\sum_{k=1}^{n}x_{k}+\mathrm{noise})\|_{\infty}+\|\int_{\{0,1\}^{p}}xx^{T}\,d\mu(x)-(\frac{1}{n}\sum_{k=1}^{n}x_{k}x_{k}^{T}+\mathrm{noise})\|_{\infty}$, where $\|\,\|_{\infty}$ is the $\ell^{\infty}$ norm on $\mathbb{R}^{p}$ or $\mathbb{R}^{p^{2}}$.
However, this requires exponential time in $p$, since the set of all probability measures on $\{0,1\}^{p}$ can be identified as a convex subset of $\mathbb{R}^{2^{p}}$. See \cite{barak2007privacy}.

If~(iii) or (iv) is dropped, the problem is trivial: in the former case, we can output either the original true data; in the latter, all zeros. 

While there are known algorithms that 
satisfy~(i)--(iii) with proofs and empirically satisfy~(iv) in simulations (see e.g., \cite{RMcKenna,JZhang,HLi,li2011provably}), the challenge is to develop an algorithm that provably satisfies all four conditions.

\medskip

Ullman and Vadhan~\cite{ullman2011pcps} showed that, assuming the existence of one-way functions, one cannot achieve (i)--(iv) even for $d=2$, if we require in (iv) that {\em all} of the $d$-dimensional marginals be preserved accurately. 
More precisely, there is no polynomial time, differentially private algorithm for generating synthetic data in $\{0,1\}^{p}$ that preserves all of the 
two-dimensional marginals with accuracy $o(1)$ if one-way functions exist.
This remarkable no-go result by Ullman and Vadhan already could put an end to our quest for finding an algorithm that rigorously can achieve conditions (i)--(iv).

Surprisingly, however, a slightly weaker interpretation of (iv) suffices to put our quest on a more successful path. Indeed, we will show in this paper that one can achieve (i)--(iv), if we require in (iv) that {\em most} of the $d$-dimensional marginals be preserved accurately. Remarkably, our result does not only hold for two-dimensional marginals, but for marginals of {\em any given degree}.
 
\medskip

Note that even if the differential privacy condition in (iii) is replaced by the condition of anonymous microaggregation, it is still a challenging open problem to develop an algorithm that fulfills all these desiderata. In this paper we will solve this problem by deriving a computationally efficient anonymous microaggregation framework that comes with provable accuracy bounds.

\medskip
{\bf Covariance loss.}
We approach the aforementioned goals by studying a fundamental, but a first glance completely unrelated, problem in probability. This problem is  concerned with the most basic notion of probability: conditional expectation. We want to answer the fundamental question: 
\begin{quote}
{\em ``How much information is lost when we take conditional expectation?''}
\end{quote}

The law of total variance shows that taking conditional expectation of a random variable underestimates the variance. A similar phenomenon holds in higher dimensions: taking conditional expectation of a random vector underestimates the covariance (in the positive-semidefinite order). We may ask: how much covariance is lost? And what sigma-algebra of given complexity minimizes the covariance loss?


Finding an answer to this fundamental probability question turns into a quest of finding among all sigma-algebras of given complexity that one which minimizes the covariance loss. We will derive  a nearly optimal bound based on a careful explicit construction of a specific sigma-algebra.
Amazingly, this excursion into theoretical probability produces mathematical techniques that are most suitable to solve the previously discussed challenging practical problems concerning microaggregation and privacy.

\subsection{Private, synthetic data?}

Now that we described the spirit of our main results, let us introduce them in more detail. 

As mentioned before, it is known from Ullman and Vadhan~\cite{ullman2011pcps} that it is generally impossible to efficiently make private synthetic data that accurately preserves all low-dimensional marginals. 
However, as we will prove, it is possible to efficiently construct private synthetic data that preserves {\em most} of the low-dimensional marginals. 

To state our goal mathematically, we average the accuracy (in the $L^2$ sense) over all $\binom{p}{d}$ subsets of indices $\{i_1,\ldots,i_d\} \subset [p]$, then take the expectation over the randomness in the algorithm. In other words, we would like to see
\begin{equation}	\label{eq: error tensor small}
\E \binom{p}{d}^{-1} \sum_{1 \le i_1 < \cdots < i_d \le p} E(i_1,\ldots,i_d)^2 \le \d^2
\end{equation}
for some small $\d$, where $E(i_1,\ldots,i_d)$ is defined in (\ref{eq: error tensor}). If this happens, we say that the synthetic data is {\em $\d$-accurate for $d$-dimensional marginals on average}. Using Markov inequality, we can see that the synthetic data is $o(1)$-accurate for $d$-dimensional marginals on average if and only if with high probability, most of the $d$-dimensional marginals are asymptotically accurate; more precisely, with probability $1-o(1)$, a $1-o(1)$ fraction of the $d$-dimensional marginals of the synthetic data is within $o(1)$ of the corresponding marginals of the true data. 

Let us state our result informally.

\begin{theorem}[Private synthetic Boolean data]			\label{thm: informal}
  Let $\e,\kappa \in (0,1)$ and $n,m \in \N$.
  There exists an $\e$-differentially private algorithm that transforms
  input data $x_1,\ldots,x_n \in \{0,1\}^p$ into output data $y_1,\ldots,y_m \in \{0,1\}^p$.   
  Moreover, if $d=O(1)$, $d \le p/2$, $m \gg 1$, $n \gg (p/\e)^{1+\kappa}$, then
  the synthetic data is $o(1)$-accurate for $d$-dimensional marginals on average. 
  The algorithm runs in time polynomial in $p$, $n$ and linear in $m$, 
  and is independent of $d$.
\end{theorem}

Theorem~\ref{thm: Boolean private} gives a formal and non-asymptotic version of this result. 

Our method is not specific to Boolean data. It can be used to generate 
synthetic data with {\em any predefined convex constraints} (Theorem~\ref{thm: unweighted privacy}).
If we assume that the input data $x_1,\ldots,x_n$ lies in some known convex set 
$K \subset \R^p$, one can make private and accurate synthetic data $y_1,\ldots,y_m$
that lies in the same set $K$. 

\subsection{Covariance loss} \label{ss:loss}

Our method is based on a new problem in probability theory, a problem that is interesting on its own. It is about the most basic notion of probability: conditional expectation. And the question is: {\em how much information is lost when we take conditional expectation}?

The law of total expectation states that for a random variable $X$ and a sigma-algebra $\FF$, the conditional expectation $Y = \E[X|\FF]$ gives an unbiased estimate of the mean:
$\E X = \E Y$.
The {\em law of total variance}, which can be expressed as 
$$
\Var(X) - \Var(Y) = \E X^2 - \E Y^2 = \E(X-Y)^2,
$$
shows that taking conditional expectation underestimates the variance.

Heuristically, the simpler the sigma-algebra $\FF$ is, the more variance gets lost. 
What is the best sigma-algebra $\FF$ with a given complexity? 
Among all sigma-algebras $\FF$ that are generated by a partition 
of the sample space into $k$ subsets, 
which one achieves the smallest loss of variance, and what is that loss?

If $X$ is bounded, let us say $\abs{X} \le 1$, we can decompose the interval $[-1,1]$ into $k$ subintervals of length $2/k$ each, take $F_i$ to be the preimage of each interval under $X$, and let $\FF = \s(F_1,\ldots,F_k)$ be the sigma-algebra generated by these events. Since $X$ and $Y$ takes values in the same subinterval a.s., we have $\abs{X-Y} \le 2/k$ a.s. Thus, the law of total variance gives 
\begin{equation}	\label{eq: variance loss}
\Var(X) - \Var(Y) \le \frac{4}{k^{2}}.
\end{equation}

Let us try to generalize this question to higher dimensions. If $X$ is a random vector taking values in $\R^p$, the law of total expectation holds unchanged. The law of total variance becomes the {\em law of total covariance}:
$$
\S_X - \S_Y = \E XX^\tran - \E YY^\tran = \E(X-Y)(X-Y)^\tran
$$
where $\S_X = \E(X-\E X)(X-\E X)^\tran$ denotes the covariance matrix of $X$, and similarly for $\Sigma_Y$ (see Lemma~\ref{lem: law of total covariance} below). Just like in the one-dimensional case, we see that taking conditional expectation underestimates the covariance (in the positive-semidefinite order).

However, if we naively attempt to bound the loss of covariance like we did to get \eqref{eq: variance loss}, we would face a curse of dimensionality. The unit Euclidean ball in $\R^p$ cannot be partitioned into $k$ subsets of diameter, let us say, $1/4$, unless $k$ is exponentially large in $p$ (see e.g.~\cite{BSS20}). 
The following theorem\footnote{The $\ell_2$ norm and the tensor notation used in this section are defined in Section~\ref{s: tensors}.}
shows that a much better bound can be obtained that does not suffer the curse of dimensionality.

\begin{theorem}[Covariance loss]		\label{thm: covariance loss}
  Let $X$ be a random vector in $\R^p$ such that $\norm{X}_2 \le 1$ a.s.
  Then, for every $k \ge 3$, there exists a partition of the sample space into at most $k$ sets
  such that for the sigma-algebra $\FF$ generated by this partition, the conditional expectation 
  $Y = \E[X|\FF]$ satisfies
  \begin{equation}\label{lossbound}
  \norm{\S_X - \S_Y}_2 
  \le C \sqrt{\frac{\log \log k}{\log k}}.
  \end{equation}
  Here $C$ is an absolute constant.
  Moreover, if the probability space has no atoms, then
  the partition can be made with exactly $k$ sets, all of which have the same probability $1/k$. 
  \end{theorem}
  
\begin{remark}[Optimality]
  The rate in Theorem~\ref{thm: covariance loss} is in general optimal 
  up to a $\sqrt{\log \log k}$ factor; 
  see Proposition~\ref{prop: optimality}.
\end{remark}

\begin{remark}[Higher moments]
  Theorem~\ref{thm: covariance loss} can be automatically extended to higher moments 
  via the following {\em tensorization principle} (Theorem~\ref{thm: higher moments}), 
  which states that for all $d \ge 2$,
  \begin{equation}	\label{eq: tensorization}
  \norm[1]{\E X^{\otimes d} - \E Y^{\otimes d}}_2
    \le 4^d \norm[1]{\E X^{\otimes 2} - \E Y^{\otimes 2}}_2
    = 4^d \norm[1]{\Sigma_X - \Sigma_Y}_2.
  \end{equation}
  \end{remark}

\begin{remark}[Hilbert spaces]
  The bound~\eqref{lossbound}   is dimension-free. 
  Indeed, Theorem~\ref{thm: covariance loss}   
  can be extended to hold for infinite dimensional Hilbert spaces.
\end{remark}

\subsection{Anonymous microaggregation}

Let us apply these abstract probability results to the problem of making synthetic data. 
As before, denote the true data by $x_1,\ldots,x_n \in \R^p$.
Let $X(i)=x_i$ be the random variable on the sample space $[n]$ equipped with uniform probability distribution. Obtain a partition $[n]=I_1 \cup \cdots \cup I_m$ from the Covariance Loss Theorem~\ref{thm: covariance loss}, where $m\leq k$, and let us assume for simplicity that $m=k$ and that all sets $I_j$ have the same cardinality $\abs{I_j} = n/k$ (this can be achieved whenever $k$ divides $n$, a requirement that can easily be dropped as we will discuss later). 
The conditional expectation $Y = \E[X|\FF]$ on the sigma-algebra $\FF = \s(I_1,\ldots,I_k)$ generated by this partition takes values 
\begin{equation}	\label{eq: microaggregation equipartition}
  y_j = \frac{k}{n} \sum_{i \in I_j} x_i, \quad j=1,\ldots,k.
\end{equation}
with probability $1/k$ each.
In other words, the synthetic data $y_1,\ldots,y_k$ is obtained by taking local averages, 
or by {\em microaggregation} of the input data $x_1,\ldots,x_n$. 
The crucial point is that the synthetic data is obviously generated via {\em $(n/k)$-anonymous microaggregation}.  Here, we use the following formal definition of $r$-anonymous microaggregation.
\begin{definition}\label{def:anonymity}
Let $x_{1},\ldots,x_{n}\in\mathbb{R}^{p}$ be a dataset. Let $r\in\mathbb{N}$. A $r$-anonymous averaging of $x_{1},\ldots,x_{n}$ is a dataset consisting of the points $\sum_{i\in I_{1}}x_{i},\ldots,\sum_{i\in I_{m}}x_{i}$ for some partition $[n]=I_{1}\cup\ldots\cup I_{m}$ such that $|I_{i}|\geq r$ for each $1\leq i\leq m$. A $r$-anonymous microaggregation algorithm ${\cal A}(\,)$ with input dataset $x_{1},\ldots,x_{n}\in\mathbb{R}^{p}$ is the composition of a $r$-anonymous averaging procedure followed by any algorithm.
\end{definition}

For any notion of privacy, any post-processing of a private dataset should still be considered as private. While a post-processing of a $r$-anonymous averaging of a dataset is not necessarily a $r$-anonymous averaging of the original dataset (it might not even consist of vectors), the notion of $r$-anonymous microaggregation allows a post-processing step after $r$-anonymous averaging.

What about the accuracy?
The law of total expectation $\E X = \E Y$ becomes
$\frac{1}{n} \sum_{i=1}^n x_i = \frac{1}{k}\sum_{j=1}^k y_j$.
As for higher moments, assume that $\norm{x_i}_2 \le 1$ for all $i$. Then Covariance Loss Theorem~\ref{thm: covariance loss} together with tensorization principle \eqref{eq: tensorization} yields
$$
\norm[3]{\frac{1}{n} \sum_{i=1}^n x_i^{\otimes d} - \frac{1}{k} \sum_{j=1}^k y_j^{\otimes d}}_2 
\lesssim 4^d \sqrt{\frac{\log \log k}{\log k}}.
$$
Thus, if $k \gg 1$ and $d=O(1)$, the synthetic data is accurate in the sense of the mean square average of marginals.

This general principle can be specialized to Boolean data. Doing appropriate rescaling, bootstrapping (Section~\ref{s: bootstrapping}) and randomized rounding (Section~\ref{s: randomized rounding}), we can conclude the following:

\begin{theorem}[Anonymous synthetic Boolean data]		\label{thm: anonymity informal}
  Suppose $k$ divides $n$.
  There exists a randomized $(n/k)$-anonymous microaggregation algorithm that transforms
  input data $x_1,\ldots,x_n \in \{0,1\}^p$ into output data 
  $y_1,\ldots,y_m \in \{0,1\}^p$.  Moreover, if $d=O(1)$, $d \le p/2$, $k \gg 1$, $m \gg 1$,  the synthetic data is $o(1)$-accurate for $d$-dimensional marginals on average. 
  The algorithm runs in time polynomial in $p$, $n$ and linear in $m$, and is independent of $d$.
\end{theorem}

Theorem~\ref{thm: Boolean anonymous} gives a formal and non-asymptotic version of this result.

\subsection{Differential privacy}
 
How can we pass from anonymity to differential privacy and establish Theorem~\ref{thm: informal}?
The microaggregation mechanism by itself is not differentially private. 
However, it {\em reduces sensitivity} of synthetic data. If a single input data point $x_i$ is changed, microaggregation \eqref{eq: microaggregation equipartition} suppresses the effect of such change on the synthetic data $y_j$ by the factor $k/n$. 
Once the data has low sensitivity, the classical {\em Laplacian mechanism} can make it private: one has simply to add Laplacian noise. 

This is the gist of the proof of Theorem~\ref{thm: informal}. However, several issues arise. One is that we do not know how to make all blocks $I_j$ of the same size while preserving their privacy, so we allow them to have arbitrary sizes in the application to privacy. However, small blocks $I_j$ may cause instability of microaggregation, and diminish its beneficial effect on sensitivity. We resolve this issue by downplaying, or {\em damping}, the small blocks (Section~\ref{s: damping}). 
The second issue is that adding Laplacian noise to the vectors $y_i$ may move them outside the set $K$ where the synthetic data must lie (for Boolean data, $K$ is the unit cube $[0,1]^p$.) We resolve this issue by metrically projecting the perturbed vectors back onto $K$ (Section~\ref{s: metric projection}). 

\subsection{Related work}

There exists a  large body of work on privately releasing answers in the interactive and non-interactive query setting. But 
a major advantage of releasing a  synthetic data set instead of just the answers to specific queries is that synthetic data opens up  a much richer toolbox (clustering, classification, regression, visualization, etc.), and thus much more flexibility, to analyze the data.

In~\cite{blum2013learning}, Blum, Ligett, and Roth gave an $\e$-differentially private synthetic data algorithm whose accuracy scales  logarithmically with the number of queries, but the complexity scales exponentially with $p$. The papers~\cite{hardt2010multiplicative,hardt2012simple} propose methods for producing private synthetic data with an error bound of about  $\tilde{\mathcal{O}}(\sqrt{n} p^{1/4})$ per query. However, the associated  algorithms have running time that is at least exponential in $p$.
In~\cite{barak2007privacy}, Barak et al.\ derive a method for producing accurate and private synthetic Boolean data based on linear programming with a running time that is exponential in $p$. This should be contrasted with the fact that our algorithm runs in time polynomial in $p$ and $n$, see Theorem~\ref{thm: anonymity informal}.

We emphasize that the our method is designed to produce synthetic data. But, as suggested by the reviewers, we briefly discuss how well $d$-way marginals can be preserved by our method in the non-synthetic data regime. Here, we consider the dependence of $n$ on $p$ as well as the accuracy we achieve versus existing methods.

{\bf Dependence of $n$ on $p$:} In order to release 1-way marginals with any nontrivial accuracy on average and with $\epsilon$-differential privacy, one already needs $n\gtrsim p$ \cite[Theorem 8.7]{dwork2014algorithmic}. Our main result Theorem \ref{thm: informal} on differential privacy only requires $n$ to grow slightly faster than linearly in $p$. 

If we just want to privately release the $d$-way marginals without creating synthetic data, 
and moreover relax $\epsilon$-differential privacy to $(\epsilon,\delta)$-differential privacy, one might relax the dependence of $n$ on $p$. Specifically, $n\gg p^{\lceil d/2\rceil/2}\sqrt{\log(1/\delta)}/\epsilon$ suffices \cite{dworknikolov}. In particular, when $d=2$, this means that $n\gg\sqrt{p\log(1/\delta)}/\epsilon$ suffices. On the other hand, our algorithm does not depend on $d$. Moreover, for $d\geq 5$, the dependence of $n$ on $p$ required in Theorem \ref{thm: informal} is less restrictive than in \cite{dworknikolov}.

{\bf Accuracy in $n$:} As mentioned above, Theorem~\ref{thm: Boolean private} gives a formal and non-asymptotic version of Theorem \ref{thm: informal}. The average error of $d$-way marginals we achieve has decay of order $\sqrt{\log\log n/\log n}$ in $n$. In Remark \ref{lowerboundaverageerror}, we show that even for $d=1,2$, no polynomial time differentially private algorithm can have average error of the marginals decay faster than $n^{-a}$ for any $a>0$, assuming the existence of one-way functions.

However, if we only need to release $d$-way marginals with differential privacy but without creating synthetic data, then one can have the average error of the $d$-way marginals decay at the rate $1/\sqrt{n}$ \cite{dworknikolov}.




\subsection{Outline of the paper}

The rest of the paper is organized as follows. In Section~\ref{s:prelim} we provide basic notation and other preliminaries. Section~\ref{s:loss} is concerned with the concept of covariance loss. We give a constructive and nearly optimal answer to the problem of how much information is lost when we take conditional expectation.
In Section~\ref{s: microaggregation} we use the tools developed for covariance loss to  derive a computationally efficient microaggregation framework that comes with provable accuracy bounds regarding low-dimensional marginals. In Section~\ref{s:privacy} we obtain analogous versions of these results in the framework of differential privacy.

\section{Preliminaries}\label{s:prelim}

\subsection{Basic notation}

The approximate inequality signs $\lesssim$ hide absolute constant factors; 
thus $a \lesssim b$ means that $a \le C b$ for a suitable absolute constant $C>0$. A list of elements $\nu_{1},\ldots,\nu_{k}$ of a metric space $M$ is an $\alpha$-covering, where $\alpha>0$, if every element of $M$ has distance less than $\alpha$ from one of $\nu_{1},\ldots,\nu_{k}$. For $p\in\mathbb{N}$, define
\[B_2^p=\{x\in\mathbb{R}^{p}:\,\|x\|_{2}\leq 1\}.\]

\subsection{Tensors}				\label{s: tensors}

The marginals of a random vector can be conveniently represented in tensor notation.
A tensor is a $d$-way array $X \in \R^{p \times \cdots \times p}$. 
In particular, $1$-way tensors are vectors, and $2$-way tensors are matrices. 
A simple example of a tensor is the rank-one tensor $x^{\otimes d}$, 
which is constructed from a vector $x \in \R^p$ by multiplying its entries:
$$
x^{\otimes d}(i_1,\ldots,i_d) = x(i_1) \cdots x(i_d), 
\quad \text{where } i_1,\ldots,i_d \in [p].
$$
In particular, the tensor $x^{\otimes 2}$ is the same as the matrix $xx^\tran$.

 The $\ell^2$ norm of a tensor $X$ can be defined by regarding $X$ as a 
vector in $\R^{p^d}$, thus
$$
\norm{X}_2^2 \coloneqq \sum_{i_1, \ldots i_d \in [p]} \abs{X(i_1,\ldots,i_d)}^2.
$$
Note that when $d=2$, the tensor $X$ can be identified as a matrix and $\|X\|_{2}$ is the Frobenius norm of $X$.

The errors of the marginals \eqref{eq: error tensor} can be thought of as the coefficients of the error tensor
\begin{equation}	\label{eq: error}
E = \frac{1}{n} \sum_{i=1}^n x_i^{\otimes d} - \frac{1}{m} \sum_{i=1}^m y_i^{\otimes d},
\end{equation}

A tensor $X \in \R^{p \times \cdots \times p}$ is {\em symmetric} if the values of its entries are independent of the permutation of the indices, i.e. if 
$$
X(i_1,\ldots,i_d) = X(i_{\pi(1)}, \ldots, i_{\pi(d)})
$$ 
for any permutation $\pi$ of $[p]$. 
It often makes sense to count each distinct entry of a symmetric tensor once instead of $d!$ times. To make this formal, we may consider the restriction operator $P_{\sym}$ 
that preserves the $\binom{p}{d}$ entries 
whose indices satisfy $1 \le i_1 < \cdots < i_d \le p$, and zeroes out all other entries. Thus 
$$
\norm{P_\sym X}_2^2 = \sum_{1 \le i_1 < \cdots < i_d \le p} X(i_1,\ldots,i_d)^2.
$$
Thus, the goal we stated in \eqref{eq: error tensor small} can be restated as follows: 
for the error tensor \eqref{eq: error}, we would like to bound the quantity
\begin{equation}	\label{eq: error as tensor sum}
\E \binom{p}{d}^{-1} \sum_{1 \le i_1 < \cdots < i_d \le p} E(i_1,\ldots,i_d)^2
= \E \binom{p}{d}^{-1} \norm{P_\sym E}_2^2.
\end{equation}

The operator $P_\sym$ is related to another restriction operator $P_\off$, 
which retains the $\binom{p}{d}d!$ off-diagonal entires, i.e. those
for which all indices $i_1, \ldots i_d$ are distinct, and zeroes out all other entries.
Thus,
\begin{equation}	\label{eq: off vs sym exact}
\norm{P_\off X}_2^2 = \sum_{i_1, \ldots i_d \in [p] \text{ distinct}} X(i_1,\ldots,i_d)^2
= d! \, \norm{P_\sym X}_2^2,
\end{equation}
for all symmetric tensor $X$.

\begin{lemma}		\label{lem: off vs sym}
  If $p \ge 2d$, we have
  \begin{equation}		\label{eq: off vs sym}
  \binom{p}{d}^{-1} \norm{P_\sym X}_2^2 
  \le \Big(\frac{2}{p}\Big)^d \norm{P_\off X}_2^2,
  \end{equation}
  for all symmetric $d$-way tensor $X$.
\end{lemma}

\begin{proof}
According to \eqref{eq: off vs sym exact}, the left hand side of \eqref{eq: off vs sym} equals 
$\big( \binom{p}{d} d! \big)^{-1} \norm{P_\off X}_2^2$, 
and $\binom{p}{d} d! = p(p-1) \cdots (p-d+1) \ge (p/2)^d$ if $p \ge 2d$.
This yields the desired bound.
\end{proof}

\subsection{Differential privacy}\label{ss:dp}

We briefly review some basic facts about differential privacy.  The interested reader may consult~\cite{dwork2014algorithmic} for details.
\begin{definition}[Differential Privacy~\cite{dwork2014algorithmic}]  A randomized function ${\mathcal M}$ gives $\epsilon$-differential privacy
 if for all databases $D_1$ and $D_2$  differing on at most one element, and all measurable $S \subseteq \range({\cal M})$, 
$$
\Pr{\mathcal M(D_1) \in S} 
\le e^{\epsilon} \cdot \Pr{\mathcal M(D_2) \in S},
$$
where the probability is with respect to the randomness of  ${\mathcal M}$.
\end{definition}

Almost all existing mechanisms to implement differential privacy are based on adding noise to the data or the data queries, e.g.\ via the  Laplacian mechanism~\cite{barak2007privacy}.  Recall that a random variable has the (centered) Laplacian distribution $\Lap(\sigma)$ if its probability density function at $x$ is
 $\frac{1}{2\sigma}\exp(-|x|/\sigma)$.

\begin{definition}\label{def:DP}
For $f: {\cal D} \to \R^d$, the $L_1$-sensitivity is
$$\Delta f: = \max_{D_1,D_2} \|f(D_1) - f(D_2)\|_1,$$
for all $D_1, D_2$ differing in at most one element.
\end{definition}

\begin{lemma}[Laplace mechanism, Theorem~2 in~\cite{barak2007privacy}]\label{le:Lap}
For any $f: {\cal D} \to \R^d$, the addition of $\Lap(\sigma)^{d}$ noise preserves $(\Delta f/\sigma)$-differential privacy.
\end{lemma}

The proof of the following lemma, which is similar in spirit to the Composition Theorem 3.14 in~\cite{dwork2014algorithmic}, is left to the reader.
\begin{lemma}\label{le:composition}
Suppose that an algorithm ${\cal A}_1: \mathcal{D} \to {\cal Y}_1$ is $\epsilon_1$-differentially private and an algorithm 
${\cal A}_2: \mathcal{D}\times{\cal Y}_1  \to {\cal Y}_2$ is $\epsilon_2$-differentially private in the first component in $\mathcal{D}$. Assume that ${\cal A}_1$ and ${\cal A}_2$ are independent. Then the composition algorithm ${\cal A} = {\cal A}_2(\,\cdot\,, {\cal A}_1(\cdot))$ is 
$(\epsilon_1+\epsilon_2)$-differentially private.
\end{lemma}

\begin{remark}\label{re:neverlookback}
As outlined in~\cite{barak2007privacy}, any function applied to private data, without accessing the raw data, is privacy-preserving.
\end{remark}

The following observation is a special case of Lemma \ref{le:composition}.
\begin{lemma}\label{le:stacking}
Suppose the data $Y_1$ and $Y_2$ are independent with respect to the randomness of the privacy-generating algorithm and that each is $\epsilon$-differentially private, then
$(Y_1,Y_2)$ is $2\epsilon$-differentially private.
\end{lemma}

\section{Covariance loss} \label{s:loss}

The goal of this section is to prove Theorem~\ref{thm: covariance loss} 
and its higher-order version, Corollary~\ref{cor: higher moments}.
We will establish the main part of Theorem~\ref{thm: covariance loss} in Sections~\ref{s: ltc}--\ref{s: covariance loss main}, the ``moreover'' part (equipartition) in Sections~\ref{s: monotonicity}--\ref{s: equipartition}, the tensorization principle \eqref{eq: tensorization} in Section~\ref{s: higher moments}, and then immediately yields Corollary~\ref{cor: higher moments}. Finally, we show optimality in Section~\ref{s: optimality}.

\subsection{Law of total covariance}		\label{s: ltc}

Throughout this section, $X$ is an arbitrary random vector in $\R^p$, 
$\FF$ is an arbitrary sigma-algebra and $Y = \E[X|\FF]$ is the conditional expectation.

\begin{lemma}[Law of total covariance]		\label{lem: law of total covariance}
  We have
  $$
  \S_X - \S_Y = \E XX^\tran - \E YY^\tran = \E(X-Y)(X-Y)^\tran.
  $$
  In particular, $\S_X \succeq \S_Y$.
\end{lemma}

\begin{proof}
The covariance matrix can be expressed as
$$
\S_X = \E(X-\E X)(X-\E X)^\tran = \E XX^\tran - (\E X)(\E X)^\tran
$$
and similarly for $Y$. Since $\E X = \E Y$ by the law of total expectation, we have 
$\S_X-\S_Y = \E XX^\tran - \E YY^\tran$, proving the first equality in the lemma.
Next, one can check that 
$$
\E \left[ XX^\tran | \FF \right] - YY^\tran = \E \left[ (X-Y)(X-Y)^\tran | \FF \right]
\quad \text{almost surely}
$$
by expanding the product in the right hand side and recalling that $Y$ is $\FF$-measurable.
Finally, take expectation on both sides to complete the proof.
\end{proof}

\begin{lemma}[Decomposing the covariance loss]		\label{lem: covariance loss two terms}
  For any orthogonal projection $P$ in $\R^p$ we have
  $$
  \norm[1]{\E XX^\tran - \E YY^\tran}_2
  \le \E \norm{PX-PY}_2^2 + \norm[1]{(I-P)(\E XX^\tran)(I-P)}_2.
  $$
\end{lemma}

\begin{proof}
By the law of total covariance (Lemma~\ref{lem: law of total covariance}),
the matrix 
$$
A \coloneqq \E XX^\tran - \E YY^\tran = \E(X-Y)(X-Y)^\tran
$$ 
is positive semidefinite. 
Then we can use the following inequality, which holds for any positive-semidefinite matrix $A$ (see e.g. in \cite[p.157]{Audenaert}):
\begin{equation}	\label{eq: A decomposition}
\norm{A}_2 \le \norm{PAP}_2 + \norm{(I-P)A(I-P)}_2.
\end{equation}

Let us bound the two terms in the right hand side. 
Jensen's inequality gives
$$
\norm{PAP}_2 
\le \E \norm[1]{P (X-Y)(X-Y)^\tran P}_2
= \E \norm{PX-PY}_2^2.
$$
Next, since the matrix $\E YY^\tran$ is positive-semidefinite, we have 
$0 \preceq A \preceq \E XX^\tran$ in the semidefinite order, 
so $0 \preceq (I-P)A(I-P) \preceq (I-P)(\E XX^\tran)(I-P)$, 
which yields 
$$
\norm{(I-P)A(I-P)}_2 \le \norm[1]{(I-P)(\E XX^\tran)(I-P)}_2.
$$
Substitute the previous two bounds into \eqref{eq: A decomposition} to complete the proof.
\end{proof}

\subsection{Spectral projection}			\label{s: spectral projection}

The two terms in Lemma~\ref{lem: covariance loss two terms} will be bounded separately. 
Let us start with the second term. It simplifies if $P$ is a spectral projection:

\begin{lemma}[Spectral projection]		\label{lem: spectral projection}
  Assume that $\norm{X}_2 \le 1$ a.s. 
  Let $t\in\mathbb{N}\cup\{0\}$. Let $P$ be the orthogonal projection in $\R^p$ 
  onto the $t$ leading eigenvectors of the second moment matrix $S = \E XX^\tran$. 
  Then 
  $$
  \norm[1]{(I-P)S(I-P)}_2
  \le \frac{1}{\sqrt{t}}.
  $$ 
\end{lemma}
  
\begin{proof}  
We have  
\begin{equation}	\label{eq: projection via eigs}
\norm[1]{(I-P)S(I-P)}_2^2 
= \sum_{i>t} \l_i(S)^2
\end{equation}
where $\l_i(S)$ denote the eigenvalues of $S$ arranged in a non-increasing order.
Using linearity of expectation and trace, we get
$$
\sum_{i=1}^p \l_i(S) 
= \E \tr XX^\tran
= \E \norm{X}_2^2
\le 1.
$$
It follows that at most $t$ eigenvalues of $S$ can be larger than $1/t$. 
By monotonicity, this yields $\l_i(S) \le 1/t$ for all $i > t$. 
Combining this with the bound above, we conclude that
\begin{equation}	\label{eq: tail eigs}
\sum_{i>t} \l_i(S)^2 
\le \sum_{i>t} \l_i(S) \cdot \frac{1}{t} 
\le \frac{1}{t}.
\end{equation}
Substitute this bound into \eqref{eq: projection via eigs} to complete the proof.
\end{proof}

\subsection{Nearest-point partition}\label{ss:nearest}

Next, we bound the first term in Lemma~\ref{lem: covariance loss two terms}. 
This is the only step that does not hold generally but for a specific sigma-algebra, 
which we generate by a nearest-point partition. 

\begin{definition}[Nearest-point partition]		\label{def: npp X}
  Let $X$ be a random vector taking values in $\R^p$, defined 
  on a probability space $(\Omega, \Sigma, \P)$. 
  A nearest-point partition $\{F_1,\ldots,F_s\}$ of $\Omega$
  with respect to a list of points $\nu_1,\ldots,\nu_s \in \R^p$ is a partition $\{F_1,\ldots,F_s\}$ of $\Omega$ such that
  \[\|\nu_j-X(\omega)\|_2=\min_{1\leq i\leq s}\|\nu_i-X(\omega)\|_2,\]
  for all $\omega\in F_j$ and $1\leq j\leq s$. (Some of the $F_j$ could be empty.)
\end{definition}

\begin{remark}
A nearest-point partition can be constructed as follows:
  For each $\omega \in \Omega$, choose a point $\nu_j$ nearest to $X(\omega)$ in the $\ell^2$ metric 
  and put $\omega$ into $F_j$. Break any ties arbitrarily as long as the $F_j$ are measurable.
\end{remark}

\begin{lemma}[Approximation]		\label{lem: approximation}
  Let $X$ be a random vector in $\R^p$ such that $\norm{X}_2 \le 1$ a.s.
  Let $P$ be an orthogonal projection on $\R^p$.
  Let $\nu_1,\ldots,\nu_s \in \R^p$ be an $\alpha$-covering of the unit Euclidean ball of $\ran(P)$.
  Let $\Omega=F_1 \cup \cdots \cup F_s$ be a nearest-point partition 
  for $PX$ with respect to $\nu_1,\ldots,\nu_s$. 
  Let $\FF = \s(F_1,\ldots,F_s)$ be the sigma-algebra generated by the partition. 
  Then the conditional expectation $Y = \E[X|\FF]$ satisfies
  $$
  \norm{PX-PY}_2 \le 2\alpha	\quad \text{almost surely.}
  $$
\end{lemma}

\begin{proof}
If $\omega \in F_j$ then, by the definition of the nearest-point partition,
$\|\nu_j-PX(\omega)\|_2=\min_{1\leq i\leq s}\|\nu_i-PX(\omega)\|_2$. So by the definition of the $\alpha$-covering, 
we have $\norm{PX(\omega)-\nu_j}_2 \le \alpha$. Hence, by the triangle inequality we have
\begin{equation}	\label{eq: diameter}
\norm{P(X(\omega)-X(\omega'))}_2 \le 2\alpha
\quad \text{whenever } \omega,\omega' \in F_j.
\end{equation}

Furthermore, by definition of $Y$, we have
$$
Y(\omega) = \frac{1}{\P(F_j)} \int_{F_j} X(\omega') \, d\P(\omega')
\quad \text{whenever } \omega \in F_j.
$$
Thus, for such $\omega$ we have
$$
X(\omega) - Y(\omega) = \frac{1}{\P(F_j)} \int_{F_j} \left( X(\omega)-X(\omega') \right) \, d\P(\omega').
$$
Applying the projection $P$ and taking norm on both sides, then using Jensen's inequality, we conclude that
$$
\norm{PX(\omega) - PY(\omega)}_2 
\le \frac{1}{\P(F_j)} \int_{F_j} \norm{P \left( X(\omega)-X(\omega') \right)}_2 \, d\P(\omega')
\le 2\alpha,
$$
where in the last step we used \eqref{eq: diameter}.
Since the bound holds for each $\omega \in F_j$ and the events $F_j$ form a partition of $\Omega$, 
it holds for all $\omega \in \Omega$. The proof is complete.
\end{proof}

\subsection{Proof of the main part of Theorem~\ref{thm: covariance loss}}		\label{s: covariance loss main}

The following simple (and possibly known) observation will come handy to bound the cardinality of an $\alpha$-covering that we will need in the proof of Theorem~\ref{thm: covariance loss}.

\begin{proposition}[Number of lattice points in a ball] \label{epsilonnet}
 For all $\alpha>0$ and $t\in\mathbb{N}$,
 $$
 \left| B_2^t \cap \frac{\alpha}{\sqrt{t}} \Z^t \right|
 \le \left( \Big( \frac{1}{\alpha}+\frac{1}{2} \Big) \sqrt{2e\pi} \right)^t.
 $$
 In particular, for any $\alpha \in (0,1)$, it follows that
 $$
 \left| B_2^t \cap \frac{\alpha}{\sqrt{t}} \Z^t \right|
 \le \left( \frac{7}{\alpha} \right)^t.
 $$
\end{proposition}

\begin{proof}
The open cubes of side length $\alpha/\sqrt{t}$ that are centered at the points of the
set $\NN \coloneqq B_2^t \cap \frac{\alpha}{\sqrt{t}} \Z^t$ are all disjoint. 
Thus the total volume of these cubes equals $\abs{\NN} (\alpha/\sqrt{t})^t$.

On the other hand, since each such cube is contained in a ball of radius $\alpha/2$ centered at some point of $\NN$, the union of these cubes is contained in the ball $(1+\alpha/2) B_2^t$.
So, comparing the volumes, we obtain
$$
\abs{\NN} (\alpha/\sqrt{t})^t \le (1+\alpha/2)^t \Vol(B_2^t),
$$
or 
$$
\abs{\NN} \le \left( \Big( \frac{1}{\alpha}+\frac{1}{2} \Big) \sqrt{t} \right)^t \Vol(B_2^t).
$$

Now, it is well known that~\cite{sommerville}
$$
\Vol(B_2^t) =\frac{\pi^{t/2}}{\Gamma(t/2+1)}.
$$
Using Stirling's formula we have
$$
\Gamma(x+1) \ge \sqrt{\pi} (x/e)^x, \quad
x \ge 0.
$$
This gives
$$
\Vol(B_2^t) \le (2e\pi/t)^{t/2}.
$$
Substitute this into the bound on $\abs{\NN}$ above to complete the proof.
\end{proof}

It follows now from Proposition~\ref{epsilonnet} that for every $\alpha \in (0,1)$, 
there exists an $\alpha$-covering in the unit Euclidean ball of dimension $t$ of cardinality 
at most $(7/\alpha)^t$. 

Fix an integer $k \ge 3$ and choose
\begin{equation}	\label{eq: t nu}
t \coloneqq \left\lfloor \frac{\log k}{\log(7/\alpha)} \right\rfloor, \quad
\alpha \coloneqq \Big( \frac{\log \log k}{\log k} \Big)^{1/4}.
\end{equation}
The choice of $t$ is made so that we can find an $\alpha$-covering of the unit Euclidean ball of $\ran(P)$ of cardinality at most $(7/\alpha)^t \le k$.

We decompose the covariance loss $\Sigma_{X}-\Sigma_{Y}$ in Lemma~\ref{lem: law of total covariance} into two terms as in Lemma~\ref{lem: covariance loss two terms}
and bound the two terms as in Lemma~\ref{lem: spectral projection}
and Lemma~\ref{lem: approximation}. This way we obtain
$$
\norm{\S_X - \S_Y}_2 
= \norm[1]{\E XX^\tran - \E YY^\tran}_2
\le 4\alpha^2 + \frac{1}{\sqrt{t}}
\lesssim \sqrt{\frac{\log \log k}{\log k}},
$$
where the last bound follows from our choice of $\alpha$ and $t$. If $t=0$ then $k\leq C$, for some universal constant $C>0$, so $\|\Sigma_X-\Sigma_Y\|_{2}$ is at most $O(1)$ and $\sqrt{\log\log k/\log k}=O(1)$. The main part of Theorem~\ref{thm: covariance loss} is proved.
\qed

\subsection{Monotonicity}				\label{s: monotonicity}

Next, we are going to prove the ``moreover'' (equipartition) part of Theorem~\ref{thm: covariance loss}. This part is crucial in the application for anonymity, but it can be skipped if the reader is only interested in differential privacy.

Before we proceed, let us first note a simple monotonicity property:

\begin{lemma}[Monotonicity]		\label{lem: monotonicity}
  Conditioning on a larger sigma-algebra can only decrease the covariance loss. 
  Specifically, if $Z$ is a random variable and $\BB \subset \GG$ are sigma-algebras
  then 
  $$
  \norm[1]{\S_Z - \S_{\E[Z|\GG]}}_2 \le \norm[1]{\S_Z - \S_{\E[Z|\BB]}}_2.
  $$  
\end{lemma}

\begin{proof}
Denoting $X = \E[Z|\GG]$ and $Y = \E[Z|\BB]$, we see from the law of total expectation that $Y = \E[X|\BB]$. The law of total covariance (Lemma~\ref{lem: law of total covariance}) then yields 
$\S_Z \succeq \S_X \succeq \S_Y$, which we can rewrite as 
$0 \preceq \S_Z-\S_X \preceq \S_Z-\S_Y$.
From this relation, it follows that $\norm{\S_Z-\S_X}_2 \le \norm{\S_Z-\S_Y}_2$, as claimed.
\end{proof}

Passing to a smaller sigma-algebra may in general increase the covariance loss. 
The additional covariance loss can be bounded as follows:

\begin{lemma}[Merger]		\label{lem: merger}
  Let $Z$ be a random vector in $\R^p$ such that $\norm{Z}_2 \le 1$ a.s.
  If a sigma-algebra is generated by a partition, merging elements of the partition 
  may increase the covariance loss by at most the total probability of the merged sets.
  Specifically, if $\GG = \s(G_1,\ldots,G_m)$ and $\BB = \s(G_1 \cup \cdots \cup G_r, G_{r+1},\ldots,G_m)$, then the random vectors $X = \E[Z|\GG]$ and $Y = \E[Z|\BB]$ satisfy
  $$
  0 \le \norm{\S_Z-\S_Y}_2 - \norm{\S_Z-\S_X}_2 
  \le \P(G_1 \cup \cdots \cup G_r).
  $$
\end{lemma}

\begin{proof}
The lower bound follows from monotonicity (Lemma~\ref{lem: monotonicity}).
To prove the upper bound, we have
\begin{equation}	\label{eq: added covariance loss}
\norm{\S_Z-\S_Y}_2 - \norm{\S_Z-\S_X}_2
\le \norm{\S_X-\S_Y}_2 
\le \E \norm{X-Y}_2^2
\end{equation}
where the first bound follows by triangle inequality, and the second 
from Lemma~\ref{lem: law of total covariance} and Lemma~\ref{lem: covariance loss two terms} for $P=I$.

Denote by $\E_G$ the conditional expectation on the set $G = G_1 \cup \cdots \cup G_r$, i.e. $\E_G[Z] = \E[Z|G] = \P(G)^{-1} \E [Z \one_G]$. Then 
$$
Y(\omega) = 
\begin{cases}
  \E_G X, & \omega \in G \\
  X(\omega), & \omega \in G^c
\end{cases}
$$
Indeed, to check the first case, note that since $\BB \subset \GG$, the law of total expectation yields $Y = \E[X|\BB]$; then the case follows since $G \in \BB$.
To check the second case, note that since the sets $G_{r+1}, \cdots,G_m$ belong to both sigma-algebras $\GG$ and $\BB$, so the conditional expectations $X$ and $Y$ 
must agree on each of these sets and thus on their union $G^c$.
Hence
$$
\E \norm{X-Y}_2^2 
= \E\norm{X-Y}_2^2 \one_G
= \P(G) \cdot \E_G \norm{X-\E_G X}_2^2
\le \P(G) \cdot \E_G \norm{X}_2^2
\le \P(G).
$$
Here we bounded the variance by the second moment,
and used the assumption that $\norm{X}_2 \le 1$ a.s.
Substitute this bound into \eqref{eq: added covariance loss} to complete the proof.
\end{proof}

\subsection{Proof of equipartition (the ``moreover'' part of Theorem~\ref{thm: covariance loss})}	\label{s: equipartition}

Let $k'  = \lfloor \sqrt{k} \rfloor$. Assume that $k'\geq 3$. (Otherwise $k<9$ and the result is trivial by taking arbitrary partition into $k$ of the same probability.) Applying the first part of Theorem~\ref{thm: covariance loss} for $k'$ instead of $k$, we obtain a sigma-algebra $\FF'$ generated by a partition of a sample space into at most $k'$ sets $F_i$, and such that
$$
\norm{\S_X - \S_{\E[X|\FF']}}_2 
\lesssim \sqrt{\frac{\log \log k'}{\log k'}}
\lesssim \sqrt{\frac{\log \log k}{\log k}}.
$$

Divide each set $F_i$ into subsets with probability $1/k$ each using division with residual. Thus we partition each $F_i$ into a certain number of subsets (if any) of probability $1/k$ each and one residual subset of probability less than $1/k$. By Monotonicity  Lemma~\ref{lem: monotonicity}, any partitioning can only reduce the covariance loss.

This process results in the creation of a lot of good subsets -- each having probability $1/k$ -- and at most $k'$ residual subsets that have probability less than $1/k$ each. Merge all residuals into one new ``residual subset''. While a merger may increase the covariance loss, Lemma~\ref{lem: merger} guarantees that the additional loss is bounded by the probability of the set being merged. 
Since we chose $k' = \lfloor \sqrt{k} \rfloor$, the probability of the residual subset is less than $k' \cdot (1/k) \le 1/\sqrt{k}$. So the additional covariance loss is bounded by $1/\sqrt{k}$.

Finally, divide the residual subset into further subsets of probability $1/k$ each. By monotonicity, any partitioning may not increase the covariance loss.
At this point we partitioned the sample space into subsets of probability $1/k$ each and one smaller residual subset. Since $k$ is an integer, the residual must have probability zero, and thus can be added to any other subset without affecting the covariance loss.

Let us summarize. We partitioned the sample space into $k$ subsets of equal probability such that the covariance loss is bounded by 
$$
\frac{1}{\sqrt{k}} + C \sqrt{\frac{\log \log k}{\log k}}
\lesssim \sqrt{\frac{\log \log k}{\log k}}.
$$ 
The proof is complete. \qed

\subsection{Higher moments: tensorization}			\label{s: higher moments}

Recall that Theorem~\ref{thm: covariance loss} provides a bound on the covariance 
loss\footnote{Recall Lemma~\ref{lem: law of total covariance} for the first identity,
and refer to Section~\ref{s: tensors} for the tensor notation.}
\begin{equation}	\label{eq: sigmas difference}
\S_X-\S_Y
= \E XX^\tran - \E YY^\tran
= \E X^{\otimes 2} - \E Y^{\otimes 2}.
\end{equation}
Perhaps counterintuitively, the bound on the covariance loss can automatically 
be lifted to higher moments, at the cost of multiplying the error by at most $4^d$.

\begin{theorem}[Tensorization]			\label{thm: higher moments}
  Let $X$ be a random vector in $\R^p$ such that $\norm{X}_2 \le 1$ a.s., 
  let $\FF$ be a sigma-algebra, and let $d \ge 2$ be an integer. 
  Then the conditional expectation $Y = \E[X|\FF]$ satisfies
  $$
  \norm[1]{\E X^{\otimes d} - \E Y^{\otimes d}}_2
  \le 2^{d-2}(2^d-d-1) \norm[1]{\E X^{\otimes 2} - \E Y^{\otimes 2}}_2.
  $$
\end{theorem}

For the proof, we need an elementary identity:

\begin{lemma}		\label{lem: UV}
  Let $U$ and $V$ be independent and identically distributed random vectors in $\R^p$. 
  Then
  $$
  \E \ip{U}{V}^2 = \norm[1]{\E UU^\tran}_2^2.
  $$
\end{lemma}

\begin{proof}
We have
\begin{align*} 
\E \ip{U}{V}^2
&= \E(V^\tran U)(U^\tran V)
= \E \tr V^\tran U U^\tran V \\
&= \E \tr UU^\tran VV^\tran
	\quad \text{(cyclic property of trace)} \\
&= \tr \E UU^\tran VV^\tran
	\quad \text{(linearity)} \\
&= \tr \E[UU^\tran] \, \E[VV^\tran]
	\quad \text{(independence)} \\
&= \tr (\E UU^\tran)^2
	\quad \text{(identical distribution)} \\
&=\norm[1]{\E UU^\tran}_2^2
	\quad \text{(the matrix $\E UU^\tran$ is symmetric)}.
\end{align*}
\end{proof}

\medskip

\begin{proof}[Proof of Theorem~\ref{thm: higher moments}]
{\em Step 1: binomial decomposition.}
Denoting
$$
X_0 = Y, \quad X_1 = X-Y,
$$
we can represent 
$$
X^{\otimes d} 
= (X_0+X_1)^{\otimes d}
=\sum_{i_1,\ldots,i_d \in \{0,1\}} X_{i_1} \otimes \cdots \otimes X_{i_d}.
$$
Since $Y^{\otimes d} = X_0 \otimes \cdots \otimes X_0$, it follows that
$$
X^{\otimes d} - Y^{\otimes d} 
= \sum_{\substack{i_1,\ldots,i_d \in \{0,1\} \\ i_1+\cdots+i_d \ge 1}} 
	X_{i_1} \otimes \cdots \otimes X_{i_d}.
$$
Taking expectation on both sides and using triangle inequality, we obtain
\begin{equation}	\label{eq: tensor sum}
\norm[1]{\E X^{\otimes d} - \E Y^{\otimes d}}_2
\le \sum_{\substack{i_1,\ldots,i_d \in \{0,1\} \\ i_1+\cdots+i_d \ge 1}} 
	\norm{\E X_{i_1} \otimes \cdots \otimes X_{i_d}}_2.
\end{equation}

Let us look at each summand on the right hand side separately.

\medskip

{\em Step 2: Dropping trivial terms.}
First, let us check that all summands for which $i_1+\cdots+i_d = 1$ vanish. 
Indeed, in this case exactly one term in the product $X_{i_1} \otimes \cdots \otimes X_{i_d}$ is $X_1$, while all other terms are $X_0$. Let $\E_\FF$ denote conditional expectation with respect to $\FF$. Since
$\E_\FF X_1 = \E_\FF[X-\E_\FF X] = 0$
and $X_0=Y=\E_\FF X$ is $\FF$-measurable, 
it follows that $\E_\FF X_{i_1} \otimes \cdots \otimes X_{i_d} = 0$. Thus, 
$\E X_{i_1} \otimes \cdots \otimes X_{i_d} = 0$ as we claimed. 

\medskip
{\em Step 3: Bounding nontrivial terms.}
Next, we bound the terms for which $r = i_1+\cdots+i_d \ge 2$.
Let $(X'_0,X'_1)$ be an independent copy of the pair of random variables $(X_0,X_1)$.
Then $\E X_{i_1} \otimes \cdots \otimes X_{i_d} = \E X'_{i_1} \otimes \cdots \otimes X'_{i_d}$, so 
\begin{align*} 
\norm{\E X_{i_1} \otimes \cdots \otimes X_{i_d}}_2^2
&= \ip{\E X_{i_1} \otimes \cdots \otimes X_{i_d}}{\E X'_{i_1} \otimes \cdots \otimes X'_{i_d}} \\
&= \E \ip{X_{i_1} \otimes \cdots \otimes X_{i_d}}{X'_{i_1} \otimes \cdots \otimes X'_{i_d}}
  	\quad \text{(by independence)} \\
&= \E \ip{X_{i_1}}{X'_{i_1}} \cdots \ip{X_{i_d}}{X'_{i_d}} \\
&= \E \ip{X_0}{X'_0}^{d-r} \ip{X_1}{X'_1}^r. 
\end{align*}
By assumption, we have $\norm{X}_2 \le 1$ a.s., which implies by Jensen's inequality that
$\norm{X_0}_2 = \norm{\E_\FF X}_2 \le \E_\FF \norm{X}_2 \le 1$ a.s. 
These bounds imply by the triangle inequality that $\norm{X_1}_2 = \norm{X-X_0}_2 \le 2$ a.s.
By identical distribution, we also have $\norm{X'_0}_2 \le 1$ and $\norm{X'_1}_2 \le 2$ a.s.
Hence, 
$$
\abs{\ip{X_0}{X'_0}} \le 1, \quad
\abs{\ip{X_1}{X'_1}} \le 4 
\quad \text{a.s.}
$$
Returning to the term we need to bound, this yields
\begin{align*} 
\norm{\E X_{i_1} \otimes \cdots \otimes X_{i_d}}_2^2
&\le 4^{r-2} \E \ip{X_1}{X'_1}^2 \\
&\le 4^{d-2} \norm[1]{\E X_1 X_1^\tran}_2^2 
  	\quad \text{(by Lemma~\ref{lem: UV})} \\
&= 4^{d-2} \norm[1]{\E(X-Y)(X-Y)^\tran}_2^2 \\
&= 4^{d-2} \norm[1]{\E X^{\otimes 2} - \E Y^{\otimes 2}}_2^2
  	\quad \text{(by Lemma~\ref{lem: law of total covariance})}.
\end{align*}

\medskip
{\em Step 4: Conclusion.}
Let us summarize. The sum on the right side of \eqref{eq: tensor sum} has $2^d-1$ terms.
The $d$ terms corresponding to $i_1+\cdots+i_d=1$ vanish. 
The remaining $2^d-d-1$ terms are bounded by 
$K \coloneqq 2^{d-2} \norm[1]{\E X^{\otimes 2} - \E Y^{\otimes 2}}_2$ each. 
Hence the entire sum is bounded by 
$(2^d-d-1) K$, as claimed. 
The theorem is proved.
\end{proof}

Combining the Covariance Loss Theorem~\ref{thm: covariance loss} with Theorem~\ref{thm: higher moments} in view of \eqref{eq: sigmas difference}, we conclude:

\begin{corollary}[Tensorization]		\label{cor: higher moments}
  Let $X$ be a random vector in $\R^p$ such that $\norm{X}_2 \le 1$ a.s.
  Then, for every $k \ge 3$, there exists a partition of the sample space into at most $k$ sets
  such that for the sigma-algebra $\FF$ generated by this partition, the conditional expectation 
  $Y = \E[X|\FF]$ satisfies for all $d\in\mathbb{N}$,
  $$
  \norm[1]{\E X^{\otimes d} - \E Y^{\otimes d}}_2 
  \lesssim 4^d \sqrt{\frac{\log \log k}{\log k}}.
  $$
  Moreover, if the probability space has no atoms, then 
  the partition can be made with exactly $k$ sets, all of which have the same probability $1/k$. 
\end{corollary}

\begin{remark}
  A similar bound can be deduced for the higher-order version of covariance matrix,  
  $\S^{(d)}_X \coloneqq \E (X-\E X)^{\otimes d}$. Indeed, applying Theorem~\ref{thm: covariance loss} and Theorem~\ref{thm: higher moments} for $X-\E X$ instead of $X$ (and so for $Y-\E Y$ instead of $Y$), we conclude that 
$$
\norm[1]{\S^{(d)}_X - \S^{(d)}_Y}_2
\le 8^d \norm[1]{\S^{(2)}_X - \S^{(2)}_Y}_2
\lesssim 8^d \sqrt{\frac{\log \log k}{\log k}}.
$$
(The extra $2^d$ factor appears because from $\norm{X}_2 \le 1$ we can only conclude that
$\norm{X-\E X}_2 \le 2$, so the bound needs to be normalized accordingly.)
\end{remark}

\subsection{Optimality}			\label{s: optimality}

The following result shows that the rate in Theorem~\ref{thm: covariance loss} is in general optimal up to a $\sqrt{\log\log k}$ factor.

\begin{proposition}[Optimality]				\label{prop: optimality}
  Let $p>16\ln(2k)$. Then there exists a random vector $X$ in $\mathbb{R}^{p}$ such that $\norm{X}_2 \le 1$ a.s.~and for any sigma-algebra $\FF$ generated by a partition of a sample space into at most $k$ sets, the conditional expectation $Y = \E[X|\FF]$ satisfies
  $$
  \norm{\S_X - \S_Y}_2
  \ge \frac{1}{80 \sqrt{\ln(2k)}}.
  $$
\end{proposition}

We will make $X$ uniformly distributed on a well-separated subset of the Boolean cube $p^{-1/2}\{0,1\}^p$ of cardinality $n=2k$. The following well known lemma states that such a subset exists:

\begin{lemma}[A separated subset]		\label{lem: separated subset}
  Let $p>16\ln n$.
  Then there exist points $x_1,\ldots,x_n \in p^{-1/2}\{0,1\}^p$ such that
  $$
  \norm{x_i-x_j}_2 > \frac{1}{2} \quad \text{for all distinct } i,j \in [n].
  $$
\end{lemma}

\begin{proof}
Let $X$ and $X'$ be independent random vectors uniformly distributed on $\{0,1\}^p$. 
Then $\norm{X-X'}_2^2 = \sum_{r=1}^p (X(r)-X'(r))^2$ is a sum of i.i.d. Bernoulli random variables with parameter $1/2$. Then Hoeffding's inequality~\cite{grimmett2020probability}  yields
$$
\Pr{\norm{X-X'}_2^2 \le p/4} \le e^{-p/8}.
$$
Let $X_1,\ldots,X_n$ be independent random vectors uniformly distributed on $\{0,1\}^p$.
Applying the above inequality for each pair of them and then taking the union bound, we conclude that
$$
\Pr{\exists i,j \in [n] \text{ distinct }: \; \norm{X_i-X_j}_2^2 \le p/4} 
\le n^2 e^{-p/8} < 1
$$
due to the condition on $n$.
Therefore, there exists a realization of these random vectors that satisfies
$$
\norm{X_i-X_j}_2 > \frac{\sqrt{p}}{2}
\quad \text{for all distinct } i,j \in [n].
$$
Divide both sides by $\sqrt{p}$ to complete the proof. 
\end{proof}

We will also need a high-dimensional version of the identity $\Var(X) = \frac{1}{2}\E(X-X')^{2}$ where $X$ and $X'$ are independent and identically distributed random variables. The following generalization is straightforward:

\begin{lemma}	\label{lem: var iid}
  Let $X$ and $X'$ be independent and identically distributed random vectors 
  taking values in $\R^p$. Then 
  $$
  \E\norm{X-\E X}_2^2 = \frac{1}{2}\E\norm{X-X'}_2^2.
  $$
\end{lemma}

\begin{proof}[Proof of Proposition~\ref{prop: optimality}]
Let $n=2k$.
Consider the sample space $[n]$ equipped with uniform probability and the sigma-algebra that consists of all subsets of $[n]$. 
Define the random variable $X$ by
$$
X(i)=x_i, \quad i \in [n]
$$
where $\{x_1,\ldots,x_n\}$ is the $(1/2)$-separated subset of $p^{-1/2}\{0,1\}^p$ from Lemma~\ref{lem: separated subset}. Hence, $X$ is uniformly distributed on the set $\{x_1,\ldots,x_n\}$.

Now, if $\FF$ is the sigma-algebra generated by a partition $\{F_1,\ldots,F_{k_0}\}$ of $[n]$ with $k_0\leq k$, then
\begin{align*} 
\sqrt{p} \norm{\S_X-\S_Y}_2
  &\ge \tr(\S_X-\S_Y)		\\
  &= \tr \E(X-Y)(X-Y)^\tran
  	\quad \text{(by Lemma~\ref{lem: law of total covariance} again)} \\
  &= \E \tr(X-Y)(X-Y)^\tran
  = \E \norm{X-Y}_2^2 \\
  &=\E \left[ \E_\FF \norm{X-\E_\FF X}_2^2 \right]
  	\quad \text{(where $\E_\FF$ denotes conditional expectation)} \\
  &=\sum_{j=1}^{k_0} \P(F_j) \; \E \norm{X_j-\E X_j}_2^2 \\
  \intertext{where the random variable $X_j$ is uniformly distributed on the set $\{x_i\}_{i \in F_j}$.}
  &=\frac{1}{2}\sum_{j=1}^{k_0} \P(F_j) \; \E \norm[1]{X_j-X'_j}_2^2
\end{align*}
where $X'_j$ is an independent copy of $X_j$, by Lemma~\ref{lem: var iid}.
Since the $X_j$ and $X'_j$ are independent and uniformly distributed on the set of $\abs{F_j}$ points, $\norm[1]{X_j-X'_j}_2$ can either be zero (if both random vectors hit the same point, which happens with probability $1/\abs{F_j}$) or it is greater than $1/2$ by separation. Hence 
$$
\E \norm[1]{X_j-X'_j}_2^2 \ge \frac{1}{4} \Big(1-\frac{1}{\abs{F_j}}\Big).
$$
Moreover, $\P(F_j) = \abs{F_j}/{n}$, so substituting in the bound above yields
$$
\sqrt{p} \norm{\S_X-\S_Y}_2
\ge \frac{1}{2}\sum_{j=1}^{k_0} \frac{\abs{F_j}}{n} \cdot \frac{1}{4} \Big(1-\frac{1}{\abs{F_j}}\Big)
= \frac{1}{8n}(n-k_0)
\geq \frac{1}{16},
$$
where we used that the sets $F_j$ form a partition of $[n]$ so their cardinalities sum to $n$, our choice of $n=2k$ and the fact that $k_0\leq k$.

We proved that 
$$
\norm{\S_X-\S_Y}_2
\ge \frac{1}{16\sqrt{p}}.
$$
If $p \le 25\ln n$, this quantity is further bounded below by $1/(80 \sqrt{\ln n}) = 1/(80 \sqrt{\ln(2k)})$, completing the proof in this range. For larger $p$, the result follows by appending enough zeros to $X$ and thus embedding it into higher dimension. Such embedding obviously does not change $\norm{\S_X-\S_Y}_2$.
\end{proof}

\section{Anonymity}			\label{s: microaggregation}

In this section, we use our results on the covariance loss to make anonymous and accurate synthetic data by microaggregation. To this end, we can interpret microaggregation probabilistically as conditional expectation (Section~\ref{s: mgg as conditional}) 
and deduce a general result on anonymous microaggregation (Theorem~\ref{thm: anonymity}). 
We then show how to make synthetic data with custom size by bootstrapping (Section~\ref{s: bootstrapping}) and Boolean synthetic data by randomized rounding (Section~\ref{s: randomized rounding}).

\subsection{Microaggregation as conditional expectation}		\label{s: mgg as conditional}
For discrete probability distributions, conditional expectation can be interpreted as microaggregation, or local averaging. 

Consider a finite sequence of points $x_1,\ldots,x_n \in \R^p$, which we can think of as true data. Define the random variable $X$ on the sample space $[n]$ equipped with the uniform probability distribution by setting
$$
X(i)=x_i, \quad i \in [n].
$$
Now, if $\FF = \s(I_1,\ldots,I_k)$ is the sigma-algebra generated by some partition
$[n]=I_1 \cup \cdots \cup I_k$, 
the conditional expectation $Y = \E[X|\FF]$ must take a constant value on each set $I_j$, 
and that value is the average of $X$ on that set. In other words, $Y$ takes values $y_j$ with probability $w_j$, where
\begin{equation}	\label{eq: microaggregation}
w_j = \frac{\abs{I_j}}{n}, \quad 
y_j = \frac{1}{\abs{I_j}} \sum_{i \in I_j} x_i,
\quad j=1,\ldots,k.
\end{equation}

The law of total expectation $\E X = \E Y$ in our case states that
\begin{equation}	\label{eq: LTE}
\frac{1}{n} \sum_{i=1}^n x_i = \sum_{j=1}^k w_j y_j.
\end{equation}
Higher moments are handled using Corollary~\ref{cor: higher moments}. 
This way, we obtain an effective anonymous algorithm that creates synthetic data while accurately preserving most marginals:

\begin{theorem}[Anonymous microaggregation]			\label{thm: anonymity}
  Suppose $k$ divides $n$. There exists a (deterministic) algorithm that 
  takes input data $x_1,\ldots,x_n \in \R^p$ such that $\norm{x_i}_2 \le 1$ for all $i$, 
  and computes a partition $[n]=I_1 \cup \cdots \cup I_k$ with $\abs{I_j} = n/k$ for all $j$,
  such that the microaggregated vectors 
  $$
  y_j = \frac{k}{n} \sum_{i \in I_j} x_i, \quad j=1,\ldots,k,
  $$
  satisfy for all $d\in\mathbb{N}$,
  $$
  \norm[3]{\frac{1}{n} \sum_{i=1}^n x_i^{\otimes d} - \frac{1}{k} \sum_{j=1}^k y_j^{\otimes d}}_2 
  \lesssim 4^d \sqrt{\frac{\log \log k}{\log k}}.
  $$
  The algorithm runs in time polynomial in $p$ and $n$,  
  and is independent of $d$. 
\end{theorem}

\begin{proof}
Most of the statement follows straightforwardly from Corollary~\ref{cor: higher moments}
in  light of the discussion above. However, the ``moreover" part of Corollary~\ref{cor: higher moments} requires the probability space to be atomless, while our probability space $[n]$ does have atoms.
Nevertheless, if the sample space consists of $n$ atoms of probability $1/n$ each, and $k$ divides $n$, then it is obvious that
the divide-and-merge argument explained in Section~\ref{s: equipartition} works, and so the ``moreover" part of Corollary~\ref{cor: higher moments} also holds in this case. Thus, we obtain the $(n/k)$-anonymity from the microaggregation procedure. It is also clear that the algorithm (which is independent of $d$) runs in time polynomial in $p$ and $n$. See the Microaggregation part of Algorithm \ref{algorithm1}.


\end{proof}

\begin{remark}
The requirement that $k$ divides $n$ appearing in Theorem~\ref{thm: anonymity} as well as in other theorems makes it possible to partition $[n]$ into $k$ sets of {\em exactly} the same cardinality. While convenient for proof purposes, this assumption  is not strictly necessary. One can drop this assumption and make one set slightly larger than others. The corresponding modifications are left to the reader.
\end{remark}

The use of spectral projection in combination with microaggregation has also been proposed in~\cite{Monedero}, although without any theoretical analysis regarding privacy or utility.

\subsection{Synthetic data with custom size: bootstrapping}			\label{s: bootstrapping}

A seeming drawback of Theorem~\ref{thm: anonymity} is that the anonymity strength $n/k$ 
and the cardinality $k$ of the output data $y_1,\ldots,y_k$ are tied to each other. To produce synthetic data of arbitrary size, we can use the classical technique of {\em bootstrapping}, 
which consists of sampling new data $u_1,\ldots,u_m$ 
from the data $y_1,\ldots,y_k$ independently and with replacement. 
The following general lemma establishes the accuracy of resampling:

\begin{lemma}[Bootstrapping]	\label{lem: bootstrapping}
  Let $Y$ be a random vector in $\R^p$ such that $\norm{Y}_2 \le 1$ a.s.  
  Let $Y_1,\ldots,Y_m$ be independent copies of $Y$. Then for all $d \in \N$ we have
  $$
  \E \norm[3]{\frac{1}{m} \sum_{i=1}^m Y_i^{\otimes d} - \E Y^{\otimes d}}_2^2
  \le \frac{1}{m}.
  $$
\end{lemma}

\begin{proof}
We have
\begin{align*} 
\E \norm[3]{\frac{1}{m} \sum_{i=1}^m Y_i^{\otimes d} - \E Y^{\otimes d}}_2^2
  &= \frac{1}{m^2} \sum_{i=1}^m \E \norm[1]{Y_i^{\otimes d} - \E Y^{\otimes d}}_2^2
  	\quad \text{(by independence and zero mean)} \\
  &= \frac{1}{m} \E \norm[1]{Y^{\otimes d} - \E Y^{\otimes d}}_2^2
  	\quad \text{(by identical distribution)} \\
  &= \frac{1}{m} \left( \E \norm[1]{Y^{\otimes d}}_2^2 - \norm[1]{\E Y^{\otimes d}}_2^2 \right)
  \le \frac{1}{m} \E \norm[1]{Y^{\otimes d}}_2^2 
  = \frac{1}{m} \E \norm{Y}_2^{2d}.
\end{align*}
Using the assumption $\norm{Y}_2 \le 1$ a.s., we complete the proof.
\end{proof}

Going back to the data $y_1,\ldots,y_k$ produced by Theorem~\ref{thm: anonymity}, 
let us consider a random vector $Y$ that takes values $y_j$ with probability $1/k$ each. 
Then obviously
$\E Y^{\otimes d} = \frac{1}{k} \sum_{j=1}^k y_j^{\otimes d}$.
Moreover, the assumption that $\norm{x_i}_2 \le 1$ for all $i$ 
implies that
$\norm{y_j}_2 \le 1$ for all $j$, 
so we have $\norm{Y}_2 \le 1$ as required in Bootstrapping Lemma~\ref{lem: bootstrapping}.
Applying this lemma, we get
$$
\E \norm[3]{\frac{1}{m} \sum_{i=1}^m Y_i^{\otimes d} - \frac{1}{k} \sum_{j=1}^k y_j^{\otimes d}}_2^2
\le \frac{1}{m}.
$$
Combining this with the bound in Theorem~\ref{thm: anonymity}, we obtain:

\begin{theorem}[Anonymous microaggregation: custom data size]		\label{thm: anonymity unlimited}
  Suppose $k$ divides $n$. Let $m\in\mathbb{N}$. There exists a randomized $(n/k)$-anonymous microaggregation algorithm that 
  transforms input data $x_1,\ldots,x_n \in B_2^p$ to the output data $u_1,\ldots, u_m \in B_2^p$ in such a way that for all $d\in\mathbb{N}$,  
  $$
  \E \norm[3]{\frac{1}{n} \sum_{i=1}^n x_i^{\otimes d} - \frac{1}{m} \sum_{i=1}^m u_i^{\otimes d}}_2^2 
  \lesssim 16^d \, \frac{\log \log k}{\log k} + \frac{1}{m}.
  $$
  The algorithm consists of anonymous averaging (described in Theorem~\ref{thm: anonymity})
  followed by bootstrapping (described above).
  It runs in time polynomial in $p$ and $n$, and is independent of $d$.
\end{theorem}

\begin{remark}[Convexity]			\label{rem: convexity}
  Microaggregation respects convexity. If the input data $x_1,\ldots,x_n$ lies in some given convex set $K$, the output data $u_1,\ldots, u_m$ will lie in $K$, too. 
  This can be useful in applications where one often needs to preserve some natural constraints on the data, such as positivity. 
\end{remark}

\subsection{Boolean data: randomized rounding}		\label{s: randomized rounding}

Let us now specialize to Boolean data. 
Suppose the input data $x_1,\ldots,x_n$ is taken from $\{0,1\}^p$.
We can use Theorem~\ref{thm: anonymity unlimited} (and obvious renormalization by the factor $\norm{x_i}_2 = \sqrt{p}$) 
to make $(n/k)$-anonymous synthetic data $u_1,\ldots,u_m$ that satisfies
\begin{equation}	\label{eq: microaggregation renormalized}
\E p^{-d} \norm[3]{\frac{1}{n} \sum_{i=1}^n x_i^{\otimes d} - \frac{1}{m} \sum_{i=1}^k u_i^{\otimes d}}_2^2 
\lesssim 16^d \, \frac{\log \log k}{\log k} + \frac{1}{m}.
\end{equation}

According to Remark~\ref{rem: convexity}, the output data $u_1,\ldots,u_m$ lies in the cube $K=[0,1]^p$.  In order to transform the vectors $u_i$ into Boolean vectors, 
i.e. points in $\{0,1\}^p$, we can apply the known technique of {\em randomized rounding}~\cite{raghavan1987randomized}. 
We define the randomized rounding of a number $x \in [0,1]$ as a random variable $r(x) \sim \Ber(x)$. Thus, to compute $r(x)$, we flip a coin that comes up heads with probability $x$ and output $1$ for a head and $0$ for a tail. It is convenient to think of $r: [0,1] \to \{0,1\}$ as a random function. The randomized rounding $r(x)$ of a vector $x \in [0,1]^p$ 
is obtained by randomized rounding on each of the $p$ coordinates of $x$ independently.

\begin{theorem}[Anonymous synthetic Boolean data]		\label{thm: Boolean anonymous}
  Suppose $k$ divides $n$.
  There exists a randomized $(n/k)$-anonymous microaggregation algorithm that transforms
  input data $x_1,\ldots,x_n \in \{0,1\}^p$ into output data 
  $z_1,\ldots,z_m \in \{0,1\}^p$ in such a way that the error 
  $E = \frac{1}{n} \sum_{i=1}^n x_i^{\otimes d} - \frac{1}{m} \sum_{i=1}^m z_i^{\otimes d}$
  satisfies 
  $$
  \E \binom{p}{d}^{-1} \sum_{1 \le i_1 < \cdots < i_d \le p} E(i_1,\ldots,i_d)^2
  \lesssim 32^d \Big( \frac{\log \log k}{\log k} + \frac{1}{m} \Big)
  $$   
  for all $d \le p/2$.
  The algorithm consists of anonymous averaging and bootstrapping (as in Theorem~\ref{thm: anonymity unlimited}) followed by 
  independent randomized rounding of all coordinates of all points. 
  It runs in time polynomial in $p$, $n$ and linear in $m$, and is independent of $d$.
\end{theorem}

For convenience of the reader, Algorithm~\ref{algorithm1} below gives a pseudocode description of the algorithm described in Theorem~\ref{thm: Boolean anonymous}.

\begin{algorithm}[h!]
\caption{Boolean $n/k$-anonymous synthetic data via microaggregation}
\label{algorithm1}
\begin{algorithmic}

\State {\bf Input:}  a  sequence of points $x_1,\ldots,x_n$ in the cube $\{0,1\}^p$  (true data);  $k \ge 9$, where $k$ divides $n$; $m\in\mathbb{N}$ (number of points in the synthetic data).

\State {\bf Microaggregation}
\begin{enumerate}
\State\indent Compute the second-moment matrix $S = \frac{1}{n} \sum_{i=1}^n x_i x_i^\tran$.
\State\indent Let $k'= \lfloor \sqrt{k} \rfloor$. Let $t \coloneqq \left\lfloor \frac{\log k'}{\log(7/\alpha)} \right\rfloor$ and $\alpha \coloneqq \Big( \frac{\log \log k'}{\log k'} \Big)^{1/4}.$
\State\indent Let $P:\mathbb{R}^{d}\to\mathbb{R}^{d}$ be the orthogonal projection onto the span of the eigenvectors
 associated with the $t$ largest eigenvalues of $S$.
\State\indent Choose an $\alpha$-covering $\nu_1,\ldots,\nu_{s} \in \R^p$ of the unit Euclidean ball of the subspace $\ran(P)$. This is done by enumerating $B_2^t\cap(\alpha/\sqrt{t}) \Z^t$ and mapping it into $\ran(P)$ using any linear isometry.
\State\indent Construct a nearest-point partition $[n] = F_1 \cup \cdots \cup F_{s}$ for $Px_{1},\ldots,Px_{n}$ with respect to $\nu_1,\ldots,\nu_{s}$ as follows.
  For each $\ell \in [n]$, choose a point $\nu_j$ nearest to $x_\ell$ in the $\ell^2$ metric 
  and put $\ell$ into $F_j$. Break any ties arbitrarily. 
\State\indent Transform the partition $[n]=F_1 \cup \cdots \cup F_{s}$ into the equipartition $[n]=I_1 \cup \cdots \cup I_k$ with $\abs{I_j} = \frac{n}{k}\,\, \forall j$ following the steps in Section~\ref{s: equipartition}: Divide each non-empty set $F_i$ into subsets with probability $1/k$ each using division with residual, then merge all residuals into one new residual subset and divide the residual subset into further subsets of probability $1/k$ each.
\State\indent Perform microaggregation: compute $y_j = \frac{k}{n} \sum_{i \in I_j} x_i, \quad j=1,\ldots,k.$
\end{enumerate}

\State {\bf Bootstrapping} creates new data $u_1,\ldots,u_m$ by sampling (independently and with replacement) $m$ points from the data $y_1,\ldots,y_k$. 
\State {\bf Randomized rounding} maps the data $\{u_\ell\}_{\ell=1}^m \in [0,1]^p$ to data $\{z_j\}_{j=1}^m \in \{0,1\}^p$.

\State {\bf Output:} a sequence of points $z_1,\ldots,z_m$  in the cube $\{0,1\}^p$ (synthetic data) that satisfy the properties outlined in 
Theorem~\ref{thm: Boolean anonymous}.
\end{algorithmic}

\end{algorithm}

To prove  Theorem~\ref{thm: Boolean anonymous}, first note:\footnote{This may be a good time for the reader to refer to Section~\ref{s: tensors} for definitions of restriction operators on tensors.}

\begin{lemma}[Randomized rounding is unbiased] \label{lem: randomized rounding}
  For any $x \in [0,1]^p$ and $d \in \N$, 
  all off-diagonal entries of the tensors $\E r(x)^{\otimes d}$ and $x^{\otimes d}$ match:
  $$
  P_\off \big( \E r(x)^{\otimes d} - x^{\otimes d} \big) = 0,
  $$
 where $P_\off$ is the orthogonal projection onto the subspace of tensors supported on the off-diagonal entries.
\end{lemma}

\begin{proof}
For any tuple of distinct indices $i_1,\ldots,i_d \in [p]$,
the definition of randomized rounding implies that
$r(x)_{i_1}, \ldots, r(x)_{i_d}$ are independent $\Ber(x_{i_1}), \ldots, \Ber(x_{i_d})$ random variables.
Thus
$$\E r(x)_{i_1} \cdots r(x)_{i_d} = x_{i_1} \cdots x_{i_d},$$
completing the proof.
\end{proof}

\medskip

\begin{proof}[Proof of Theorem~\ref{thm: Boolean anonymous}]
Condition on the data $u_1,\ldots,u_m$ obtained in Theorem~\ref{thm: anonymity unlimited}.
The output data of our algorithm can be written as $z_i = r_i(u_i)$,
where the index $i$ in $r_i$ indicates that we perform randomized rounding on each point 
$u_i$ independently.
Let us bound the error introduced by randomized rounding, which is
$$
a \coloneqq 
\E p^{-d} \norm[3]{P_\off \Big( \frac{1}{m} \sum_{i=1}^m z_i^{\otimes d} - \frac{1}{m} \sum_{i=1}^m u_i^{\otimes d} \Big)}_2^2
= \frac{p^{-d}}{m^2} \E \norm[3]{\sum_{i=1}^m Z_i}_2^2
$$
where $Z_i \coloneqq P_\off \big( r_i(u_i)^{\otimes d} - u_i^{\otimes d} \big)$
are independent mean zero random variables due to Lemma~\ref{lem: randomized rounding}.
Therefore, 
$$
a = \frac{p^{-d}}{m^2} \sum_{i=1}^m \E \norm{Z_i}_2^2.
$$
Since the variance is bounded by the second moment, we have
$$
\E \norm{Z_i}_2^2 
\le \E \norm{P_\off \big( r_i(u_i)^{\otimes d} \big)}_2^2
\le \E \norm{r_i(u_i)^{\otimes d}}_2^2
= \E \norm{r_i(u_i)}_2^{2d}
\le p^d
$$
since $r_i(u_i) \in \{0,1\}^p$. 
Hence 
$$
a \le \frac{1}{m}.
$$
Lifting the conditional expectation (i.e. taking expectation with respect to $u_1,\ldots,u_m$) and combining this with \eqref{eq: microaggregation renormalized} via triangle inequality, we obtain
$$
\E p^{-d} \norm[3]{P_\off \Big( \frac{1}{n} \sum_{i=1}^n x_i^{\otimes d} - \frac{1}{m} \sum_{i=1}^m z_i^{\otimes d} \Big)}_2^2
\lesssim 16^d \, \frac{\log \log k}{\log k} + \frac{2}{m}.
$$
Finally, we can replace the off-diagonal norm by the symmetric norm using Lemma~\ref{lem: off vs sym}. If $p \ge 2d$, it yields
$$
\E \binom{p}{d}^{-1} \norm[3]{P_\sym \Big( \frac{1}{n} \sum_{i=1}^n x_i^{\otimes d} - \frac{1}{m} \sum_{i=1}^m z_i^{\otimes d} \Big)}_2^2
\lesssim 2^d \Big( 16^d \, \frac{\log \log k}{\log k} + \frac{2}{m} \Big).
$$
In view of \eqref{eq: error as tensor sum}, the proof is complete.
\end{proof}

\section{Differential Privacy} \label{s:privacy}

Here we pass from anonymity to differential privacy by noisy microaggregation. In Section \ref{s: privproj}, we construct a ``private" version of the PCA projection using repeated applications of the ``exponential mechanism'' \cite{kapralov2013}. This ``private" projection is needed to make the PCA step in Algorithm \ref{algorithm1} in Section \ref{s: randomized rounding} differentially private.
In Sections~\ref{s: perturbing S}--\ref{s: accuracy}, we show that the microaggregation is sufficiently stable with respect to additive noise, as long as we damp small blocks $F_j$ (Section~\ref{s: damping}) and project the weights $w_j$ and the vectors $y_j$ 
back to the unit simplex and the convex set $K$, respectively (Section~\ref{s: metric projection}).
We then establish differential privacy in Section~\ref{s: privacy analyzed} and accuracy in Section~\ref{s: accuracy}, with Theorem~\ref{thm: weighted privacy} being the most general result on private synthetic data. Just like we did for anonymity, we then show how to make synthetic data with custom size by bootstrapping (Section~\ref{s: bootstrapping privacy}) and Boolean synthetic data by randomized rounding (Section~\ref{s: randomized rounding privacy}).

\subsection{Differentially private projection}\label{s: privproj}
If $A$ is a self-adjoint linear transformation on a real inner product space, then the $i$th largest eigenvalue of $A$ is denoted by $\lambda_{i}(A)$; the spectral norm of $A$ is denoted by $\|A\|$; and the Frobenius norm of $A$ is denoted by $\|A\|_{2}$. If $v_{1},\ldots,v_{t}\in\mathbb{R}^{p}$ then $P_{v_{1},\ldots,v_{t}}$ denotes the orthogonal projection from $\mathbb{R}^{p}$ onto $\mathrm{span}\{v_{1},\ldots,v_{t}\}$. In particular, if $v\in\mathbb{R}^{p}$ then $P_{v}$ denotes the orthogonal projection from $\mathbb{R}^{p}$ onto $\mathrm{span}\{v\}$.

In this section, we construct for any given $p\times p$ positive semidefinite $A$ and $1\leq t\leq p$, a random projection $P$ that behaves like the projection onto the $t$ leading eigenvectors of $A$ and, at the same time, is ``differentially private" in the sense that if $A$ is perturbed a little, the distribution of $P$ changes a little. Something like this is done \cite{kapralov2013}. However, in \cite{kapralov2013}, a PCA approximation of $A$ is produced in the output rather than the projection. The error in the operator norm for this approximation is estimated in \cite{kapralov2013}, whereas in this paper, we need to estimate the error in the Frobenius norm.

Thus, we will do a modification of the construction in \cite{kapralov2013}. But the general idea is the same: first construct a vector that behaves like the principal eigenvector (i.e., 1-dimensional PCA) and, at the same time, is ``differentially private." Repeatedly doing this procedure gives a ``differentially private" version of the $t$-dimensional PCA projection.

The following algorithm is referred to as the ``exponential mechanism" in \cite{kapralov2013}. As shown in Lemma \ref{kapralovresult} below, this algorithm outputs a random vector that behaves like the principal eigenvector (see part 1) and is ``differentially private" in the sense of part 3.
\begin{algorithm}
\caption{$\quad\mathrm{PVEC}(A)$}
\label{algorithmexpmech}
\begin{algorithmic}

\State {\bf Input:}  positive semidefinite linear transformation $A:V\to V$, where $V$ is a finite dimensional real inner product space.

\State {\bf Output:} $x$ sampled from the unit sphere of $V$ according to the density proportional to $e^{\langle Ax,x\rangle}$.
\end{algorithmic}

\end{algorithm}

\begin{lemma}[\cite{kapralov2013}]\label{kapralovresult}
Suppose that $A$ is a positive semidefinite linear transformation on a finite dimensional vector space $V$.
\begin{enumerate}[(1)]
\item If $v$ is an output of $\mathrm{PVEC}(A)$, then
$$
\E \ip{Av}{v} \ge (1-\gamma) \l_1(A)
$$
for all $\gamma>0$ such that $\lambda_{1}(A)\geq C \dim(V)\frac{1}{\gamma}\log(\frac{1}{\gamma})$, where $C>0$ is an absolute constant.
\item $\mathrm{PVEC}(A)$ can be implemented in time $\mathrm{poly}(\mathrm{dim}\, V,\lambda_{1}(A))$.
\item Let $B: V \to V$ be a positive semidefinite linear transformation.
If $\|A-B\|\leq\beta$, then 
$$
\Pr{\mathrm{PVEC}(A)\in\mathcal{S}}
\leq e^{\beta} \cdot \Pr{\mathrm{PVEC}(B)\in\mathcal{S}}
$$
for any measurable subset $\mathcal{S}$ of $V$.
\end{enumerate}
\end{lemma}

Let us restate part~1 of Lemma~\ref{kapralovresult} more conveniently: 

\begin{lemma}\label{squarediscrepancy}
 Suppose that $A$ is a positive semidefinite linear transformations on a finite dimensional vector space $V$. If $v$ is an output of $\mathrm{PVEC}(A)$, then
$$
\lambda_{1}(A)^{2} - \E \ip{Av}{v}^2
\leq 2\gamma\lambda_{1}(A)^{2}+C\dim(V)^{2}\frac{1}{\gamma^{2}}\log^{2}\left(\frac{1}{\gamma}\right)
$$
for all $\gamma>0$, where $C>0$ is an absolute constant.
\end{lemma}
\begin{proof}
Fix $\gamma>0$, and let us consider two cases. 

If $\lambda_1(A) \ge C\dim(V)\frac{1}{\gamma}\log(\frac{1}{\gamma})$, then by part~1 of Lemma~\ref{kapralovresult}, we have 
$\lambda_1(A) - \E \ip{Av}{v} \le \gamma \lambda_{1}(A)$.
Therefore, keeping in mind that the inequality $\ip{Av}{v} \le \l_1(A)$ always hold, we obtain 
$\lambda_{1}(A)^{2}- \E \ip{Av}{v}^2
= \E \left[ (\lambda_{1}(A)+\langle Av,v\rangle)(\lambda_{1}(A)-\langle Av,v\rangle) \right]\leq2\gamma\lambda_{1}(A)^{2}$.

If $\lambda_1(A) < C\dim(V)\frac{1}{\gamma}\log(\frac{1}{\gamma})$, then
$\lambda_1(A)^{2}- \E \ip{Av}{v}^2
\le \lambda_1(A)^{2}
\leq C^{2}(\mathrm{dim}\,V)^{2}\frac{1}{\gamma^{2}}\log^{2}(\frac{1}{\gamma})$.
The proof is complete.
\end{proof}

We now construct a ``differentially private" version of the $t$-dimensional PCA projection. This is done by repeated applications of PVEC in Algorithm \ref{algorithmexpmech}.
\begin{algorithm}
\caption{$\quad\mathrm{PROJ}(A,t)$}
\label{algorithmdpproj}
\begin{algorithmic}

\State {\bf Input:}  $p\times p$ positive semidefinite real matrix $A$; and $1\leq t\leq p$

\State Apply $\mathrm{PVEC}(A)$ to obtain $v_{1}\in\mathbb{R}^{p}$ with $\|v_{1}\|_{2}=1$
\State {\bf for} $i=1,\ldots,t-1$,\\
consider the linear transformation $A_i = (I-P_{v_{1},\ldots,v_{i}})A(I-P_{v_{1},\ldots,v_{i}})$ on the space $V_i = \mathrm{ran}(I-P_{v_{1},\ldots,v_{i}})$;\\
Apply $\mathrm{PVEC}(A_i)$ to obtain $v_{i+1} \in V_i$ with $\|v_{i+1}\|_{2}=1$.\\
{\bf end for}

\State {\bf Output:} Orthogonal projection $P_{v_{1},\ldots,v_{t}}$ on $\mathbb{R}^{p}$.
\end{algorithmic}

\end{algorithm}

The following lemma shows that the algorithm PROJ is ``differentially private" in the sense of part 3 of Lemma \ref{kapralovresult}, except that $e^{\beta}$ is replaced by $e^{t\beta}$.
\begin{lemma}\label{PROJprivate}
Suppose that $A$ and $B$ are $p\times p$ positive semidefinite matrices and $1\leq t\leq p$. If $\|A-B\|\leq\beta$ then
$$
\Pr{\mathrm{PROJ}(A,t)\in\mathcal{S}} 
\leq e^{t\beta} \cdot \Pr{\mathrm{PROJ}(B,t)\in\mathcal{S}}
$$
for any measurable subset $\mathcal{S}$ of $\mathbb{R}^{p\times p}$.
\end{lemma}

\begin{proof}
Fix $\beta$. We first define a notion of privacy similar to the one in \cite{kapralov2013}. A randomized algorithm $\mathcal{M}$ with input being a $p\times p$ positive semidefinite real matrix $A$ is $\theta$-DP if whenever $\|A-B\|\leq\beta$, we have $\mathbb{P}\{\mathcal{M}(A)\in\mathcal{S}\}\leq e^{\theta}\mathbb{P}\{\mathcal{M}(B)\in\mathcal{S}\}$ for all measurable subset $\mathcal{S}$ of $\mathbb{R}^{p\times p}$. 

In the algorithm PROJ, the computation of $v_{1}$ as an algorithm is $\beta$-DP by Lemma \ref{kapralovresult}(3).

Similarly, if we fix $v_{1}$, the computation of $v_{2}$ as an algorithm is also $\beta$-DP. So by some version of Lemma \ref{le:composition}, the computation of $(v_{1},v_{2})$ as an algorithm (without fixing $v_{1}$) is $2\beta$-DP.

And so on. By induction, we have that the computation of $(v_{1},\ldots,v_{t})$ as an algorithm is $t\beta$-DP. Thus, $\mathrm{PROJ}(\cdot,t)$ is $t\beta$-DP. The result follows.
\end{proof}

Next, we show that the output of the algorithm PROJ behaves like the $t$-dimensional PCA projection in the sense of Lemma \ref{PROJaccuracy} below. Observe that if $P$ is the projection onto the $t$ leading eigenvectors of a $p\times p$ positive semidefinite matrix $A$, then $\|(I-P)A(I-P)\|_{2}^{2}=\sum_{i=t+1}^{p}\lambda_{i}(A)^{2}$. To prove Lemma \ref{PROJaccuracy}, we first prove the following lemma and then we apply this lemma repeatedly to obtain Lemma \ref{PROJaccuracy}.

\begin{lemma}\label{tailcompression}
Let $A$ be a $p\times p$ positive semidefinite matrix. Let $v\in\mathbb{R}^{p}$ with $\|v\|_{2}=1$. Then
\[\sum_{i=j}^{p}\lambda_{i}((I-P_{v})A(I-P_{v}))^{2}\leq\sum_{i=j+1}^{p}\lambda_{i}(A)^{2}+\lambda_{1}(A)^{2}-\langle Av,v\rangle^{2},\]
for every $1\leq j\leq p$.
\end{lemma}
\begin{proof}
For every $p\times p$ real symmetric matrix $B$ and every $1\leq i\leq p$, we have
\[\lambda_{i}(B)=\inf_{\mathrm{dim}\,W=p-i+1}\sup_{x\in W,\,\|x\|_{2}=1}\langle Bx,x\rangle,\]
where the infimum is over all subspaces $W$ of $\mathbb{R}^{p}$ with dimension $p-i+1$. Thus, since $P_{v}$ is a rank-one orthogonal projection,
\begin{eqnarray*}
\lambda_{i}((I-P_{v})A(I-P_{v}))
&=&\inf_{\mathrm{dim}\,W=p-i+1}\sup_{x\in W,\,\|x\|_{2}=1}\langle A(I-P_{v})x,(I-P_{v})x\rangle\\&\geq&
\inf_{\mathrm{dim}\,W=p-i+1}\sup_{x\in W\cap\mathrm{ran}(I-P_{v}),\,\|x\|_{2}=1}\langle Ax,x\rangle\geq
\lambda_{i+1}(A),
\end{eqnarray*}
for every $1\leq i\leq p-1$. Thus,
\[\sum_{i=1}^{j-1}\lambda_{i}((I-P_{v})A(I-P_{v}))^{2}\geq\sum_{i=2}^{j}\lambda_{i}(A)^{2},\]
so
\begin{eqnarray*}
\sum_{i=j}^{p}\lambda_{i}((I-P_{v})A(I-P_{v}))^{2}\leq\sum_{i=j+1}^{p}\lambda_{i}(A)^{2}+\lambda_{1}(A)^{2}+\|(I-P_{v})A(I-P_{v})\|_{2}^{2}-\|A\|_{2}^{2}.
\end{eqnarray*}
Since $\|A\|_{2}^{2}-\|(I-P_{v})A(I-P_{v})\|_{2}^{2}\geq\|P_{v}AP_{v}\|_{2}^{2}=\langle Av,v\rangle^{2}$, the result follows.
\end{proof}
\begin{lemma}\label{PROJaccuracy}
Suppose that $A$ is a $p\times p$ positive semidefinite matrix and $1\leq t\leq p$. If $P$ is an output of $\mathrm{PROJ}(A,t)$, then
\[\mathbb{E}\|(I-P)A(I-P)\|_{2}^{2}\leq\sum_{i=t+1}^{p}\lambda_{i}(A)^{2}+2t\gamma\|A\|^{2}+Ct\frac{p^{2}}{\gamma^{2}}\log^{2}\left(\frac{1}{\gamma}\right),\]
for all $\gamma>0$, where $C>0$ is an absolute constant.
\end{lemma}
\begin{proof}
Let $v_{1},\ldots,v_{t}$ be those vectors defined in the algorithm $\mathrm{PROJ}(A,t)$. Let $A_{0}=A$. For $1\leq k\leq t$, let $A_{k}=(I-P_{v_{1},\ldots,v_{k}})A(I-P_{v_{1},\ldots,v_{k}})$. Since $v_{k+1}$ is an output of $\mathrm{PVEC}(A_{k})$, by Lemma~\ref{squarediscrepancy}, we have
\[\lambda_{1}(A_{k})^{2}-\mathbb{E}_{v_{k+1}}(\langle A_{k}v_{k+1},v_{k+1}\rangle^{2})\leq 2\gamma\lambda_{1}(A_{k})^{2}+C\frac{p^{2}}{\gamma^{2}}\log^{2}\left(\frac{1}{\gamma}\right),\]
for all $1\leq j\leq p$ and $0\leq k\leq t-1$, where the expectation $\mathbb{E}_{v_{k+1}}$ is over $v_{k+1}$ conditioning on $v_{1},\ldots,v_{k}$. By Lemma \ref{tailcompression}, we have
\[\sum_{i=j}^{p}\lambda_{i}((I-P_{v_{k+1}})A_{k}(I-P_{v_{k+1}}))^{2}\leq\sum_{i=j+1}^{p}\lambda_{i}(A_{k})^{2}+\lambda_{1}(A_{k})^{2}-\langle A_{k}v_{k+1},v_{k+1}\rangle^{2},\]
for all $1\leq j\leq p$ and $0\leq k\leq t-1$. Therefore,
\[\mathbb{E}_{v_{k+1}}\sum_{i=j}^{p}\lambda_{i}((I-P_{v_{k+1}})A_{k}(I-P_{v_{k+1}}))^{2}\leq\sum_{i=j+1}^{p}\lambda_{i}(A_{k})^{2}+2\gamma\lambda_{1}(A_{k})^{2}+C\frac{p^{2}}{\gamma^{2}}\log^{2}\left(\frac{1}{\gamma}\right),\]
for all $1\leq j\leq p$ and $0\leq k\leq t-1$.

In the algorithm $\mathrm{PROJ}(A,t)$, each $v_{k+1}$ is chosen from the unit sphere of $\mathrm{ran}(I-P_{v_{1},\ldots,v_{k}})$. Hence, the vectors $v_{1},\ldots,v_{t}$ are orthonormal, so $I-P_{v_{1},\ldots,v_{k+1}}=(I-P_{v_{k+1}})(I-P_{v_{1},\ldots,v_{k}})$ for all $1\leq k\leq t-1$. Thus, $(I-P_{v_{k+1}})A_{k}(I-P_{v_{k+1}})=A_{k+1}$. So we have
\[\mathbb{E}_{v_{k+1}}\sum_{i=j}^{p}\lambda_{i}(A_{k+1})^{2}\leq\sum_{i=j+1}^{p}\lambda_{i}(A_{k})^{2}+2\gamma\lambda_{1}(A_{k})^{2}+C\frac{p^{2}}{\gamma^{2}}\log^{2}\left(\frac{1}{\gamma}\right),\]
for all $1\leq j\leq p$ and $0\leq k\leq t-1$. Taking the full expectation $\mathbb{E}$ on both sides, we get
\[\mathbb{E}\sum_{i=j}^{p}\lambda_{i}(A_{k+1})^{2}\leq\mathbb{E}\sum_{i=j+1}^{p}\lambda_{i}(A_{k})^{2}+2\gamma\|A\|^{2}+C\frac{p^{2}}{\gamma^{2}}\log^{2}\left(\frac{1}{\gamma}\right),\]
for all $1\leq j\leq p$ and $0\leq k\leq t-1$, where we used the fact that $\lambda_{1}(A_{k})=\|A_{k}\|\leq\|A\|$.
Repeated applications of this inequality yields
\[\mathbb{E}\sum_{i=1}^{p}\lambda_{i}(A_{t})^{2}\leq\sum_{i=t+1}^{p}\lambda_{i}(A_{0})^{2}+2t\gamma\|A\|^{2}+Ct\frac{p^{2}}{\gamma^{2}}\log^{2}\left(\frac{1}{\gamma}\right).\]
Note that $A_{t}=(I-P_{v_{1},\ldots,v_{t}})A(I-P_{v_{1},\ldots,v_{t}})$ and $P=P_{v_{1},\ldots,v_{t}}$ is the output of $\mathrm{PROJ}(A,t)$. Thus, the left hand side is equal to $\mathbb{E}\|A_{t}\|_{2}^{2}=\mathbb{E}\|(I-P)A(I-P)\|_{2}^{2}$. The result follows.
\end{proof}

\subsection{Microaggregation with more control}			\label{s: perturbing S}

We will protect privacy by adding noise to the microaggregation mechanism.
To make this happen, we will need a version of Theorem~\ref{thm: anonymity} 
with more control. 

We adapt the microaggregation mechanism from~\eqref{eq: microaggregation} to the current setting. 
Given  a partition $[n]=F_1 \cup \cdots \cup F_s$ (where some $F_j$ could be empty), we define for $1\leq j\leq s$ with $F_j$ being non-empty,
\begin{equation}	\label{eq: microaggregation1}
w_j = \frac{\abs{F_j}}{n}, \quad 
y_j = \frac{1}{\abs{F_j}} \sum_{i \in F_j} x_i;
\end{equation}
and when $F_j$ is empty, set $w_j=0$ and $y_j$ to be an arbitrary point.

\begin{theorem}[Microaggregation with more control]			\label{thm: microaggregation noisy second moment}
  Let $x_1,\ldots,x_n \in \R^p$ be such that $\norm{x_i}_2 \le 1$ for all $i$.
  Let $S = \frac{1}{n} \sum_{i=1}^n x_i x_i^\tran$. Let $P$ be an orthogonal projection on $\R^p$.
  Let $\nu_1,\ldots,\nu_s \in \R^p$ be an $\alpha$-covering of the unit Euclidean ball of $\ran(P)$.
  Let $[n]=F_1 \cup \cdots \cup F_s$ be a nearest-point partition of $(Px_i)$ with respect to $\nu_1,\ldots,\nu_s$. 
  Then the weights $w_j$ and vectors $y_j$ defined in \eqref{eq: microaggregation1}
  satisfy for all $d \in \N$:
  \begin{equation}		\label{eq: microaggregation noisy second moment}
  \norm[3]{\frac{1}{n} \sum_{i=1}^n x_i^{\otimes d} - \sum_{j=1}^s w_j y_j^{\otimes d}}_2 
  \le 4^d \left( 4\alpha^2 + \|(I-P)S(I-P)\|_{2} \right).
  \end{equation}
\end{theorem}

\begin{proof}[Proof of Theorem~\ref{thm: microaggregation noisy second moment}]
We explained in Section~\ref{s: mgg as conditional} how to realize microaggregation probabilistically as conditional expectation. 
To reiterate, we consider the sample space $[n]$ equipped with the uniform probability distribution and define a random variable $X$ on $[n]$ by setting $X(i)=x_i$ for $i=1,\ldots,n$.
If $\FF = \s(F_1,\ldots,F_s)$ is the sigma-algebra generated by some partition
$[n]=F_1 \cup \cdots \cup F_s$, 
the conditional expectation $Y = \E[X|\FF]$ is a random vector that 
takes values $y_j$ with probability $w_j$ as defined in \eqref{eq: microaggregation1}.
Then the left hand side of \eqref{eq: microaggregation noisy second moment} equals
\begin{align*} 
\norm[1]{\E X^{\otimes d} - \E Y^{\otimes d}}_2
  &\le 4^d \norm[1]{\E XX^\tran - \E YY^\tran}_2
  	\quad \text{(by the Higher Moment Theorem~\ref{thm: higher moments})} \\
  &\le 4^d \Big( \E \norm{PX-PY}_2^2 + \norm[1]{(I-P)S(I-P)}_2 \Big)
  	\quad \text{(by Lemma~\ref{lem: covariance loss two terms})}\\
  	&\le 4^d \Big( 4\alpha^2 + \norm[1]{(I-P)S(I-P)}_2 \Big)
  	\quad \text{(by Lemma~\ref{lem: approximation})}.
\end{align*}
\end{proof}

\subsection{Perturbing the weights and vectors}		\label{s: perturbing weights and vectors}

Theorem~\ref{thm: microaggregation noisy second moment} makes the first step towards noisy microaggregation. Next, we will add noise to the weights $(w_j)$ and vectors $(y_j)$ 
obtained by microaggregation.
To control the effect of such noise on the accuracy, the following two simple bounds will be useful.

\begin{lemma}		\label{lem: tensor difference}
  Let $u,v \in \R^n$ be such that $\norm{u}_2 \le 1$ and $\norm{v}_2 \le 1$. 
  Then, for every $d \in \N$, 
  $$
  \norm[1]{u^{\otimes d} - v^{\otimes d}}_2 \le d \norm{u-v}_2.
  $$
\end{lemma}

\begin{proof}
For $d=1$ the result is trivial. For $d \ge 2$, we can represent the difference as a telescopic sum 
$$
u^{\otimes d} - v^{\otimes d}
  = \sum_{k=0}^{d-1} \left( u^{\otimes (d-k)} \otimes v^{\otimes k} - u^{\otimes (d-k-1)} \otimes v^{\otimes (k+1)} \right) 
  = \sum_{k=0}^{d-1} u^{\otimes (d-k-1)} \otimes (u-v) \otimes v^{\otimes k}. 
$$
Then, by triangle inequality,
$$
\norm[1]{u^{\otimes d} - v^{\otimes d}}_2
\le \sum_{k=0}^{d-1} \norm{u}_2^{d-k-1} \norm{u-v}_2 \norm{v}_2^{k}
\le d \norm{u-v}_2,
$$
where we used the assumption on the norms of $u$ and $v$ in the last step. 
The lemma is proved.
\end{proof}

\begin{lemma}		\label{lem: sum perturbation}
  Consider numbers $\l_j, \mu_j \in \R$ and vectors $u_j, v_j \in \R^p$ 
  such that $\norm{u_j}_2 \le 1$ and $\norm{v_j}_2 \le 1$ for all $j=1,\ldots,m$.
  Then, for every $d \in \N$, 
  \begin{equation}			\label{eq: two sums difference}
  \norm[3]{\sum_j \left( \l_j u_j^{\otimes d} - \mu_j v_j^{\otimes d} \right)}_2
  \le d \sum_j \abs{\l_j} \norm{u_j-v_j}_2 + \sum_j \abs{\l_j-\mu_j}.
  \end{equation}
\end{lemma}

\begin{proof}
Adding and subtracting the cross term $\sum_j \l_j v_j^{\otimes d}$ and using triangle inequality,  we can bound the left side of \eqref{eq: two sums difference} by 
$$
\norm[3]{\sum_j \l_j \left( u_j^{\otimes d} - v_j^{\otimes d} \right)}_2 
	+ \norm[3]{\sum_j (\l_j-\mu_j) v_j^{\otimes d}}_2
\le \sum_j \abs{\l_j} \norm[1]{u_j^{\otimes d} - v_j^{\otimes d}}_2 
	+ \sum_j \abs{\l_j-\mu_j} \norm[1]{v_j^{\otimes d}}_2.
$$
It remains to use Lemma~\ref{lem: tensor difference} and note that 
$\norm[1]{v_j^{\otimes d}}_2 = \norm[1]{v_j}_2^d$. 
\end{proof}

\subsection{Damping}			\label{s: damping}

Although the microaggregation mechanism \eqref{eq: microaggregation1} is stable with respect to additive noise in the weights $w_j$ or the vectors $y_j$ as shown in Section~\ref{s: perturbing weights and vectors}, there are still two issues that need to be resolved. 

The first issue is the potential instability of the microaggregation mechanism \eqref{eq: microaggregation1} for small blocks $F_j$. For example, if $\abs{F_j}=1$, the microaggregation does not do anything for that block and returns the original input vector $y_j=x_i$. To protect the privacy of such vector, a lot of noise is needed, which might be harmful to the accuracy.

 
 One may wonder why can we not make all blocks $F_j$ of the same size like we did in Theorem~\ref{thm: anonymity}. Indeed, in Section~\ref{s: equipartition} we showed how to transform a potentially imbalanced partition $[n]=F_1 \cup \cdots \cup F_s$ into an {\em equipartition} (where all $F_j$ have the same cardinality) using a divide-an-merge procedure; could we not apply it here? Unfortunately, an equipartition might be too sensitive\footnote{The divide-and-merge procedure described in Section \ref{s: equipartition}/Proof of Theorem \ref{thm: anonymity} for producing the $I$ blocks from the $F$ blocks is sensitive to even a change in a single data point $x_i$. For example, suppose that one block $F_1=\{x_1,\ldots,x_n\}$ contains all the points and it is divided into $I_1=\{x_1,\ldots,x_{n/k}\}$,$\cdots$, $I_{k}=\{x_{n-n/k+1},\ldots,x_n\}$. If $x_1$ is changed to another point so that it becomes a new point in another block $F_{2}$, then the new $I$ blocks could become $I_1=\{x_2,\ldots,x_{n/k+1}\}$,$\cdots$, $I_{k-1}=\{x_{n-2n/k+2},\ldots,x_{n-n/k+1}\}$, $I_{k}=\{x_{n-n/k+2},\ldots,x_n,x_1\}$ and so every $I$ block is changed by two points.} to changes even in a single data point $x_i$.
The original partition $F_1,\ldots,F_s$, on the other hand, is sufficiently stable.

We resolve this issue by suppressing, or {\em damping}, the blocks $F_j$ that are too small. 
Whenever the cardinality of $F_j$ drops below a predefined level $b$, we divide by $b$
rather than $\abs{F_j}$ in \eqref{eq: microaggregation1}. In other words, instead of vanilla microaggregation  \eqref{eq: microaggregation1}, we consider the following damped microaggregation:
\begin{equation}	\label{eq: microaggregation damping}
w_j = \frac{\abs{F_j}}{n}, \quad 
\tilde{y}_j = \frac{1}{\max(\abs{F_j},b)} \sum_{i \in F_j} x_i,
\quad j=1,\ldots,s.
\end{equation}

\subsection{Metric projection}			\label{s: metric projection}

And here is the second issue. 
Recall that the numbers $w_j$ returned by microaggregation \eqref{eq: microaggregation damping} are probability weights: the weight vector $w = (w_j)_{j=1}^s$ belongs to the unit simplex 
$$
\Delta \coloneqq \Big\{ a=(a_1,\ldots,a_s) :\; \sum_{i=1}^s a_i = 1; \; a_i \ge 0 \; \forall i \Big\}.
$$
This feature may be lost if we add noise to $w_j$. 
Similarly, if the input vectors $x_j$ are taken from a given convex set $K$ (for Boolean data, this is $K = [0,1]^p$), we would like the synthetic data to belong to $K$, too. Microaggregation mechanism \eqref{eq: microaggregation1} respects this feature: by convexity, the vectors $y_j$ do belong to $K$. However, this property may be lost if we add noise to $y_j$.

We resolve this issue by projecting the perturbed weights and vectors back onto the unit simplex $\Delta$ and the convex set $K$, respectively. For this purpose, we utilize {\em metric projections} mappings that return a proximal point in a given set. Formally, we let 
\begin{equation}\label{piprojection}
\pi_{\Delta,1}(w) \coloneqq \argmin_{\bar{w} \in \Delta} \norm{\bar{w}-w}_1; \quad
\pi_{K,2}(y) \coloneqq \argmin_{\bar{y} \in K} \norm{\bar{y}-y}_2.
\end{equation}
(If the minimum is not unique, break the tie arbitrarily. One valid choice of $\pi_{\Delta,1}(w)$ can be defined by setting all the negative entries of $w$ to be $0$ and then normalize it so that it is in $\Delta$. In the case when all entries of $w$ are negative, set $\pi_{\Delta,1}(w)$ to be any point in $\Delta$.)

Thus, here is our plan: given input data $(x_i)_{i=1}^s$, we apply damped microaggregation 
\eqref{eq: microaggregation damping} to compute weights and vectors $(w_j, \tilde{y}_j)_{j=1}^s$, add noise, and project the noisy vectors back to the unit simplex $\Delta$ and the convex set $K$ respectively. In other words, we compute
\begin{equation}	\label{eq: noise and projection}
\bar{w} = \pi_{\Delta,1} \left( w+\rho \right), \quad
\bar{y}_j = \pi_{K,2} \left( \tilde{y}_j+r_j \right),
\end{equation}
where $\rho \in \R^s$ and $r_j \in \R^p$ are noise vectors (which we will set to be random Laplacian noise in the future).

\subsection{The accuracy guarantee}			\label{s: accuracy}

Here is the accuracy guarantee of our procedure. This is a version of Theorem~\ref{thm: microaggregation noisy second moment} with noise, damping, and metric projection:

\begin{theorem}[Accuracy of damped, noisy microaggregation]	\label{thm: microaggregation perturbed}
  Let $K$ be a convex set in $\R^p$ that lies in the unit Euclidean ball $B_2^p$.
  Let $x_1,\ldots,x_n \in K$. 
  Let $S = \frac{1}{n} \sum_{i=1}^n x_i x_i^\tran$. Let $P$ be an orthogonal projection on $\R^p$.
  Let $\nu_1,\ldots,\nu_s \in \R^p$ be an $\alpha$-covering of the unit Euclidean ball of $\ran(P)$.
  Let $[n]=F_1 \cup \cdots \cup F_s$ be a nearest-point partition of $(Px_i)$ with respect to $\nu_1,\ldots,\nu_s$. 
  Then the weights $\bar{w}_j$ and vectors $\bar{y}_j$ defined in \eqref{eq: noise and projection} satisfy for all $d \in \N$:
  \begin{equation}		\label{eq: microaggregation perturbed}
  \norm[3]{\frac{1}{n} \sum_{i=1}^n x_i^{\otimes d} - \sum_{j=1}^s\bar{w}_j \bar{y}_j^{\otimes d}}_2 
  \le 4^d \left( 4\alpha^2 + \|(I-P)S(I-P)\|_{2} \right)
  	+ \frac{2dsb}{n} + 2\norm{\rho}_1 + 2d\sum_{j=1}^s w_j \norm{r_j}_2.
  \end{equation}
\end{theorem}

\begin{proof}
Adding and subtracting the cross term $\sum_j w_j y_j^{\otimes d}$ and using triangle inequality, we can bound the left hand side of \eqref{eq: microaggregation perturbed} by
\begin{equation}	\label{eq: two norms}
\norm[3]{\frac{1}{n} \sum_{i=1}^n x_i^{\otimes d} - \sum_{j=1}^s w_j y_j^{\otimes d}}_2 
	+ \norm[3]{\sum_{j=1}^s \left( w_j y_j^{\otimes d} - \bar{w}_j \bar{y}_j^{\otimes d} \right)}_2 
\end{equation}
The first term can be bounded by Theorem~\ref{thm: microaggregation noisy second moment}. 
For the second term we can use Lemma~\ref{lem: sum perturbation} and note that
\begin{equation}	\label{eq: less than 1}
\norm{y_j}_2 \le 1, \quad \norm{\bar{y}_j}_2 \le 1
\quad \text{for all } j \in [s].
\end{equation}
Indeed, definition \eqref{eq: microaggregation1} of $y_j$ and the assumption that $x_i$ lie in the convex set $K$ imply that $y_j \in K$.  Also, definition \eqref{eq: noise and projection}
implies that $\bar{y}_j \in K$ as well. Now the bounds in \eqref{eq: less than 1} follow from the 
assumption that $K \subset B_2^p$.
So, applying Theorem~\ref{thm: microaggregation noisy second moment} and Lemma~\ref{lem: sum perturbation}, 
we see that the quantity in \eqref{eq: two norms} is bounded by 
\begin{equation}	\label{eq: three terms}
4^d \left( 4\alpha^2 + \|(I-P)S(I-P)\|_{2} \right) 
+ d \sum_{j=1}^s w_j \norm{y_j-\bar{y}_j}_2 + \sum_{j=1}^s \abs{w_j-\bar{w}_j}.
\end{equation}
We bound the two sums in this expression separately.

Let us start with the sum involving $y_j$ and $\bar{y_j}$.
We will handle large and small blocks differently. 
For a large block, one for which $\abs{F_j} \ge b$, 
by \eqref{eq: microaggregation damping} we have
$\tilde{y}_j = \abs{F_j}^{-1} \sum_{i \in F_j} x_i = y_j \in K$. 
By definition \eqref{eq: noise and projection}, 
$\bar{y}_j$ is the closest point in $K$ to $y_j+r_j$. 
Since $y_j \in K$, we have
\begin{align*} 
\norm{y_j-\bar{y}_j}_2
  &\le \norm{y_j+r_j-\bar{y}_j}_2 + \norm{r_j}_2 
  	\quad \text{(by triangle inequality)} \\
  &\le \norm{y_j+r_j-y_j}_2 + \norm{r_j}_2 
  	\quad \text{(by minimality property of $\bar{y}_j$)} \\
  &= 2\norm{r_j}_2.
\end{align*}
Hence
$$
\sum_{j \in [s]:\; \abs{F_j} \ge b} w_j \norm{y_j-\bar{y}_j}_2
\le 2\sum_{j=1}^s w_j \norm{r_j}_2.
$$
Now let us handle small blocks. By \eqref{eq: less than 1} we have 
$\norm{y_j-\bar{y}_j}_2 \le 2$, so 
$$
\sum_{j \in [s]:\; \abs{F_j} < b} w_j \norm{y_j-\bar{y}_j}_2
\le \sum_{j \in [s]:\; \abs{F_j} < b} \frac{\abs{F_j}}{n} \cdot 2
\le \frac{2sb}{n}.
$$
Combining our bounds for large and small blocks, we conclude that
\begin{equation}	\label{eq: first sum}
\sum_{j=1}^s w_j \norm{y_j-\bar{y}_j}_2
\le 2\sum_{j=1}^s w_j \norm{r_j}_2 + \frac{2sb}{n}.
\end{equation}

Finally, let us bound the last sum in \eqref{eq: three terms}.
By definition \eqref{eq: noise and projection}, $\bar{w}$
is a closest point in the unit simplex $\Delta$ to $w+\rho$ in the $\ell^1$ metric.
Since $w \in \Delta$, we have
\begin{align} 
\sum_{j=1}^s \abs{w_j-\bar{w}_j}
  &= \norm{w-\bar{w}}_1 \nonumber\\
  &\le \norm{w + \rho - \bar{w}}_1 + \norm{\rho}_1 
      	\quad \text{(by triangle inequality)} \nonumber\\
  &\le \norm{w + \rho - w}_1 + \norm{\rho}_1 
      	\quad \text{(by minimality property of $\bar{w}$)} \nonumber\\  
  &= 2\norm{\rho}_1. \label{eq: second sum}
\end{align}

Substitute \eqref{eq: first sum} and \eqref{eq: second sum} into \eqref{eq: three terms} to complete the proof.
\end{proof}

\subsection{Privacy}				\label{s: privacy analyzed}

Now that we analyzed the accuracy of the synthetic data, we prove differential privacy. To that end, we will use Laplacian mechanism, so we need to bound the sensitivity of the microaggregation. 

\begin{lemma}[Sensitivity of damped microaggregation]			\label{lem: damped sensitivity}
  Let $\norm{\cdot}$ be a norm on $\R^p$.
  Consider vectors $x_1,\ldots,x_n \in \R^p$.  
  Let $I$ and $I'$ be subsets of $[n]$ that differ in exactly one element. 
  Then, for any $b>0$, we have
  \begin{equation}	\label{eq: microaggregate sensitivity}
  \norm[3]{\frac{1}{\max(\abs{I},b)} \sum_{i \in I} x_i - \frac{1}{\max(\abs{I'},b)} \sum_{i \in I'} x_i}
  \le \frac{2}{b} \max_{i \in [n]} \norm{x_i}.
  \end{equation}
\end{lemma}

\begin{proof}
Without loss of generality, we can assume that $I' = I \setminus \{n_0\}$ for some $n_0 \in I$.

\subsubsection*{Case 1: $\abs{I} \ge b+1$} 
In this case, $\abs{I'} = \abs{I}-1 \ge b$. Denoting by $\xi$ the difference vector whose norm 
we are estimating in \eqref{eq: microaggregate sensitivity}, we have
$$
\xi = \frac{1}{\abs{I}} \sum_{i \in I} x_i - \frac{1}{\abs{I}-1} \sum_{i \in I \setminus \{n_0\}} x_i
= \frac{1}{\abs{I} \left( \abs{I}-1 \right)} \sum_{i \in I \setminus \{n_0\}} (x_{n_0}-x_i).
$$
The sum in the right hand side consists of $\abs{I}-1$ terms, each satisfying 
$\norm{x_{n_0}-x_i} \le 2 \max_i \norm{x_i}$. This yields
$\norm{\xi} \le (2/\abs{I}) \max_i \norm{x_i}$. Since $\abs{I} \ge b+1$ by assumption, 
we get even a better bound than we need in this case. 

\subsubsection*{Case 2: $\abs{I} \le b$}
In this case, $\abs{I'} = \abs{I}-1 < b$. Hence the difference vector of interest equals
$$
\xi = \frac{1}{b} \sum_{i \in I} x_i - \frac{1}{b} \sum_{i \in I \setminus \{n_0\}} x_i
= \frac{x_{n_0}}{b}.
$$
Therefore, $\norm{\xi} \le (1/b) \max_i \norm{x_i}$. The lemma is proved.
\end{proof}

\begin{lemma}[Sensitivity of damped microaggregation II]			\label{lem: damped sensitivity2}
  Let $\norm{\cdot}$ be a norm on $\R^p$. Let $I$ be a subset of $[n]$ and let $n_0\in I$.
  Consider vectors $x_1,\ldots,x_n \in \R^p$ and $x_1',\ldots,x_n' \in \R^p$ such that $x_i=x_i'$ for all $i\neq n_0$.
  Then, for any $b>0$, we have
  $$
  \norm[3]{\frac{1}{\max(\abs{I},b)} \sum_{i \in I} x_i - \frac{1}{\max(\abs{I},b)} \sum_{i \in I} x_i'}
  \le \frac{1}{b}\norm{x_{n_0}-x_{n_0}'}.
  $$
\end{lemma}
\begin{proof}
$$
\norm[3]{\frac{1}{\max(\abs{I},b)} \sum_{i \in I} x_i - \frac{1}{\max(\abs{I},b)} \sum_{i \in I} x_i'}=
\norm[3]{\frac{1}{\max(\abs{I},b)}(x_{n_0}-x_{n_0}')}
  \le \frac{1}{b}\norm{x_{n_0}-x_{n_0}'}.
$$
\end{proof}

\begin{theorem}[Privacy]			\label{thm: privacy}
  In the situation of Theorem~\ref{thm: microaggregation perturbed},
  suppose that all coordinates of the vectors $\rho$ and $r_j$ 
  are independent Laplacian random variables, namely
  $$
  \rho_i \sim \Lap \Big( \frac{6}{n\e} \Big) \text{ for }i\in[s];\quad
  r_{ji} \sim \Lap \Big( \frac{12\sqrt{p}}{b\e} \Big) \text{ for }i\in[p],\;j\in[s],
  $$
  and $P$ is an output of $\mathrm{PROJ}(\frac{n\epsilon}{6t}S,t)$ where $S=\frac{1}{n}\sum_{i=1}^{n}x_{i}x_{i}^{T}$. Then the output data $(\bar{w}_j, \bar{y}_j)_{j=1}^s$ 
  is $\e$-differentially private in the input data $(x_i)_{i=1}^n$. 
\end{theorem}

\begin{proof}
First we check that the projection $P$ is private. 
To do this, let us bound the sensitivity of the second moment matrix $S = \frac{1}{n} \sum_{i=1}^n x_ix_i^{T}$ in the spectral norm.
Consider two input data $(x_i)_{i=1}^n$ and $(x'_i)_{i=1}^n$ that differ in exactly one element, 
i.e. $x_i = x'_i$ for all $i$ except some $i=n_0$. Then the difference in the spectral norm of the corresponding matrices $S$ and $S'$ satisfy
\begin{align*} 
\norm{S-S'}
  &= \frac{1}{n} \norm{x_{n_0}x_{n_0}^{T} - x'_{n_0} (x'_{n_0})^{T}} \\
  &\le \frac{1}{n} \left( \norm{x_{n_0}}_2^2 + \norm{x'_{n_0}}_2^2 \right)
  	\quad \text{(by triangle inequality)} \\
  &\le \frac{2}{n}
  	\quad \text{(since all $x_i \in K \subset B_2^p$).}
\end{align*}
Thus, $\|\frac{n\epsilon}{6t}S-\frac{n\epsilon}{6t}S'\|\leq\frac{\epsilon}{3t}$. So by Lemma \ref{PROJprivate}, the projection $P$ is $(\e/3)$-differentially private. 

Due to  Lemma~\ref{le:composition}, it suffices to prove 
that for any fixed projection $P$, the output data $(\bar{w}_j, \bar{y}_j)_{j=1}^s$ 
is $(2\e/3)$-differentially private in the input data $(x_i)_{i=1}^n$. 
Fixing $P$ fixes also the covering $(\nu_j)_{j=1}^s$.

Consider what happens if we change exactly one vector in the input data $(x_i)_{i=1}^n$.
The effect of that change on the nearest-point partition $[n]=F_1 \cup \cdots \cup F_s$ is minimal: 
at most one of the indices can move from one block $F_j$ to another block (thereby changing the cardinalities of those two blocks by $1$ each) or to another point in the same block, and the rest of the blocks 
stay the same. Thus, the weight vector $w = (w_j)_{j=1}^s$, $w_j = \abs{F_j}/n$, 
can change by at most $2/n$ in the $\ell^1$ norm. 
Due to the choice of $\rho$, it follows by Lemma~\ref{le:Lap} that $w+\rho$ is $(\e/3)$-differentially private. 

For the same reason, all vectors $\tilde{y}_j$ defined in (\ref{eq: microaggregation damping}), except for at most two, stay the same. 
Moreover, by Lemma~\ref{lem: damped sensitivity} and Lemma~\ref{lem: damped sensitivity2}, the change of each of these (at most) two vectors
in the $\ell^1$ norm is bounded by 
$$
\frac{2}{b} \max_{i \in [n]} \norm{x_i}_1 
\le \frac{2\sqrt{p}}{b} \max_{i \in [n]} \norm{x_i}_2
\le \frac{2\sqrt{p}}{b}
$$
since all $x_i \in K \subset B_2^p$.
Hence, the change of the tuple $(\tilde{y}_1,\ldots,\tilde{y}_s) \in \R^{ps}$ in the $\ell^1$ norm is bounded by $4\sqrt{p}/b$.
Due to the choice of $r_j$, it follows by Lemma~\ref{le:Lap}
that $(\tilde{y}_j+r_j)_{j=1}^s$ is $(\e/3)$-differentially private.

Since $\rho$ and $r_j$ are all independent vectors, it follows by  Lemma~\ref{le:stacking} that the pair $(w+\rho, (\tilde{y}_j+r_j)_{j=1}^s)$ is $(2\e/3)$-differentially private. 
The output data $(\bar{w}_j, \bar{y}_j)_{j=1}^s$ is a function of that pair, so it follows by Remark~\ref{re:neverlookback} that for any fixed projection $P$, that the output data must be $(2\e/3)$-differentially private. Applying Lemma~\ref{le:composition}, the result follows.
\end{proof}

\subsection{Accuracy}				\label{s: accuracy analyzed}

We are ready to combine privacy and accuracy guarantees provided by Theorem~\ref{thm: privacy} and Theorem~\ref{thm: microaggregation perturbed}.

Choose the noises $\rho \in \R^s$ and $r_j \in \R^p$ 
as in the Privacy Theorem~\ref{thm: privacy}; then 
\begin{equation}\label{rhorestimate}
\left( \E \norm{\rho}_1^2 \right)^{1/2} \lesssim \frac{s}{n\e}; \quad
\left( \E \norm{r_j}_2^2 \right)^{1/2} \lesssim \frac{p}{b\e} \text{ for all }j\in[s].
\end{equation}
To check the first bound, use triangle inequality as follows:
$$
\left( \E \norm{\rho}_1^2 \right)^{1/2} = \left( \E (\abs{\rho_1}+\cdots+\abs{\rho_s})^{2} \right)^{1/2}
\le (\E|\rho_1|^{2})^{\frac{1}{2}} + \cdots +  (\E|\rho_s|^{2})^{\frac{1}{2}},
$$
which is the sum of the standard deviations of the Laplacian distribution. 

The second bound follows from summing the variances of the Laplace distribution over all entries.

Choose $P$ to be an output of $\mathrm{PROJ}(\frac{n\epsilon}{6t}S,t)$ as in Theorem~\ref{thm: privacy}, where $S=\frac{1}{n}\sum_{i=1}^{n}x_{i}x_{i}^{T}$. Take $A=\frac{n\epsilon}{6t}S$ in Lemma \ref{PROJaccuracy}. We obtain
\[\mathbb{E}\|(I-P)A(I-P)\|_{2}^{2}\leq\sum_{i=t+1}^{p}\lambda_{i}(A)^{2}+2t\gamma\|A\|^{2}+Ct\frac{p^{2}}{\gamma^{2}}\log^{2}\left(\frac{1}{\gamma}\right),\]
and so
\[\mathbb{E}\|(I-P)S(I-P)\|_{2}^{2}\leq\sum_{i=t+1}^{p}\lambda_{i}(S)^{2}+2t\gamma\|S\|^{2}+C\frac{t^{3}p^{2}}{n^{2}\epsilon^{2}\gamma^{2}}\log^{2}\left(\frac{1}{\gamma}\right),\]
for all $\gamma>0$, where $C>0$ is an absolute constant. Since $x_{1},\ldots,x_{n}\in K\subset B_{2}^{p}$, by (\ref{eq: tail eigs}), we have $\sum_{i=t+1}^{p}\lambda_{i}(S)^{2}\leq\frac{1}{t}$. We also have $\|S\|\leq 1$. Take $\gamma=\frac{1}{t^{2}}$. We have
\begin{equation}\label{Stailbound}
\mathbb{E}\|(I-P)S(I-P)\|_{2}^{2}\leq\frac{3}{t}+C\frac{t^{7}p^{2}}{n^{2}\epsilon^{2}}\log^{2}t,
\end{equation}
where $C>0$ is an absolute constant.

Choose the accuracy $\alpha$ of the covering and its dimension $t$ as follows:
\begin{equation*}	\label{eq: net accuracy dimension}
t \coloneqq \left\lfloor\frac{\kappa\log n}{\log(7/\alpha)} \right\rfloor; \quad
\alpha = \frac{1}{(\log n)^{1/4}},
\end{equation*}
where $\kappa\in(0,1)$ is a fixed constant that will be introduced later. (See Theorem \ref{thm: weighted privacy}.)

By Proposition \ref{epsilonnet}, there exists an $\alpha$-covering in a unit ball of dimension $t$ of cardinality $s$, where
\begin{equation*}	\label{eq: net cardinality}
s \le \Big( \frac{7}{\alpha} \Big)^t \le n^{\kappa}.
\end{equation*}
Since $t\sim\frac{\kappa\log n}{\log\log n}$, from (\ref{Stailbound}), we have
\begin{equation}\label{Stailbound2}
\mathbb{E}\|(I-P)S(I-P)\|_{2}^{2}\lesssim\frac{\log\log n}{\kappa\log n}+\frac{p^{2}(\kappa\log n)^{7}}{n^{2}\epsilon^{2}},
\end{equation}
assuming that $t\geq 1$. In the case, when $t=0$, we have $n\leq C$, for some universal constant $C>0$, so the left hand side is at most $\|S\|_{2}\leq 1$ and the right hand side is at least $O(1)$.

Apply the Accuracy Theorem~\ref{thm: microaggregation perturbed} for this choice of parameters, square both sides and take expectation. Use (\ref{rhorestimate}) and (\ref{Stailbound2}).
Since the weights $w_j = \abs{F_j}/n$ satisfy $\sum_{j=1}^s w_j = 1$, we have $(\mathbb{E}(\sum_{j=1}^{s}w_{j}\|r_{j}\|_{2})^{2})^{\frac{1}{2}}\leq(\mathbb{E}\sum_{j=1}^{s}w_{j}\|r_{j}\|_{2}^{2})^{\frac{1}{2}}\lesssim\frac{p}{b\epsilon}$. The error bound in the theorem becomes
\begin{align*}
&E \coloneqq \left( \E \norm[3]{\frac{1}{n} \sum_{i=1}^n x_i^{\otimes d} - \sum_{j=1}^s \bar{w}_j \bar{y}_j^{\otimes d}}_2^2 \right)^{1/2} 
\\\lesssim&4^d \left( 4\alpha^2 + \|(I-P)S(I-P)\|_{2} \right)
  	+ \frac{2dsb}{n} + 2(\E\norm{\rho}_1^{2})^{\frac{1}{2}} + 2d\big(\E\big(\sum_{j=1}^s w_j \norm{r_j}_2\big)^{2}\big)^{\frac{1}{2}}.
  \\\lesssim&4^d \left( \frac{1}{\sqrt{\log n}} + \sqrt{\frac{\log \log n}{\kappa\log n}} + \frac{p(\kappa\log n)^{\frac{7}{2}}}{n\epsilon} \right)
  +\frac{d\cdot b\cdot n^{\kappa}}{n}+\frac{n^{\kappa}}{n\epsilon}+\frac{dp}{b\epsilon}.
\end{align*}

Optimizing $b$ leads to the following choice:
\begin{equation}	\label{eq: b}
  b = \sqrt{\frac{pn^{1-\kappa}}{\e}},
\end{equation}
and with this choice we can simplify the error bound as follows:
$$
E \lesssim 4^d \left( \sqrt{\frac{\log \log n}{\kappa\log n}} + \frac{p(\kappa\log n)^{\frac{7}{2}}}{n\epsilon} \right)
  	+ d \sqrt{\frac{p}{n^{1-\kappa}\e}} + \frac{1}{n^{1-\kappa}\e}.
$$
Note that $\kappa\log n=\frac{7}{2}\log (n^{2\kappa/7})\leq\frac{7}{2}n^{2\kappa/7}$. So $\frac{p(\kappa\log n)^{\frac{7}{2}}}{n\epsilon}\lesssim\frac{pn^{\kappa}}{n\epsilon}=\frac{p}{n^{1-\kappa}\epsilon}$.
Thus, in the range where $n \ge (p/\e)^{1/(1-\kappa)}$ for $\e \in (0,1)$, where $0<\kappa\leq1$, we have $\frac{p}{n^{1-\kappa}\e}\leq 1$, so
the error can finally be simplified to
$$
E \lesssim 4^d \left( \sqrt{\frac{\log \log n}{\kappa\log n}} + \sqrt{\frac{p}{n^{1-\kappa}\e}} \right).
$$

Note that in the complement range where $n < (p/\e)^{1/(1-\kappa)}$, the second term is greater than one, 
so such error bound is trivial to achieve by outputting $\bar{y}_j$ to be an arbitrary point in $K$ for all $j$. Thus we proved: 

\begin{theorem}[Privacy and accuracy]		\label{thm: weighted privacy}
  Let $K$ be a convex set in $\R^p$ that lies in the unit ball $B_2^p$, and $\e \in (0,1)$. Fix $\kappa\in(0,1)$.
  There exists an $\e$-differentially private algorithm that transforms
  input data $(x_i)_{i=1}^n$ where all $x_i \in K$ into the output data 
  $(\bar{w}_j, \bar{y}_j)_{j=1}^s$ where $s\leq n$, 
  all $\bar{w}_j \geq 0$, $\sum_j \bar{w}_j=1$, 
  and all $\bar{y}_j \in K$, in such a way that for all $d\in \N$:
  $$
  \E \norm[3]{\frac{1}{n} \sum_{i=1}^n x_i^{\otimes d} - \sum_{j=1}^s \bar{w}_j \bar{y}_j^{\otimes d}}_2^2
  \lesssim 16^d \left( \frac{\log \log n}{\kappa\log n} +
  \frac{p}{n^{1-\kappa}\e} \right).
  $$  
  The algorithm runs in time polynomial in $p$, $n$ 
  and linear in the time to compute the metric projection onto $K$, 
  and it is independent of $d$.
\end{theorem}

\subsection{Bootstrapping}			\label{s: bootstrapping privacy}

To get rid of the weights $\bar{w}_j$ and make the synthetic data to have custom size, 
we can use bootstrapping introduced in Section~\ref{s: bootstrapping}, 
i.e., we can sample new data $u_1,\ldots,u_m$ independently and with replacement
by choosing $\bar{y}_j$ with probability $\bar{w}_j$ at every step.

Thus, we consider the random vector $Y$ that takes value $\bar{y}_j$ with probability $\bar{w}_j$.
Let $Y_1,\ldots,Y_m$ be independent copies of $Y$. Then obviously
$\E Y^{\otimes d} = \sum_{j=1}^s \bar{w}_j \bar{y}_j^{\otimes d}$, so
Bootstrapping Lemma~\ref{lem: bootstrapping} yields
$$
\E \norm[3]{\frac{1}{m} \sum_{i=1}^m Y_i^{\otimes d} - \sum_{j=1}^s \bar{w}_j \bar{y}_j^{\otimes d}}_2^2
\le \frac{1}{m}.
$$
Combining this with the bound in Theorem~\ref{thm: weighted privacy}, we obtain:

\begin{theorem}[Privacy and accuracy: custom data size]		\label{thm: unweighted privacy}
  Let $K$ be a convex set in $\R^p$ that lies in the unit ball $B_2^p$, and $\e \in (0,1)$. Fix $\kappa\in(0,1)$
  There exists an $\e$-differentially private algorithm that transforms
  input data $x_1,\ldots,x_n \in K$ into the output data 
  $u_1,\ldots,u_m \in K$, in such a way that for all $d\in \N$:
  $$
  \E \norm[3]{\frac{1}{n} \sum_{i=1}^n x_i^{\otimes d} - \frac{1}{m} \sum_{i=1}^m u_i^{\otimes d}}_2^2
  \lesssim 16^d \left( \frac{\log \log n}{\kappa\log n} +\frac{p}{n^{1-\kappa}\e} \right) + \frac{1}{m}.
  $$  
  The algorithm runs in time polynomial in $p$, $n$ 
  and linear in $m$ and the time to compute the metric projection onto $K$, 
  and it is independent of $d$.
\end{theorem}

\subsection{Boolean data: randomized rounding}	\label{s: randomized rounding privacy}

Now we specialize to Boolean data, i.e., data from $\{0,1\}^p$. 
If the input data $x_1,\ldots,x_n$ is Boolean, the output data $u_1,\ldots,u_m$ is in $[0,1]^p$ (for technical reasons we may need to rescale the data by 
$\sqrt{p}$ because Theorem~\ref{thm: unweighted privacy} requires $K$ to be in $B_2^p$).
To transform it to Boolean data, we can use randomized rounding as described in Section~\ref{s: randomized rounding}. Thus, each coefficient of each vector $u_i$ is independently and randomly rounded to $1$ with probability equal to that coefficient, 
(and to $0$ with the complementary probability). Exactly the same analysis as we did in Section~\ref{s: randomized rounding} applies here, and we conclude:

\begin{theorem}[Boolean private synthetic data]		\label{thm: Boolean private}
  Let $\e,\kappa \in (0,1)$. There exists an $\e$-differentially private algorithm that transforms
  input data $x_1,\ldots,x_n \in \{0,1\}^p$ into the output data 
  $z_1,\ldots,z_m \in \{0,1\}^p$ in such a way that the error 
  $E = \frac{1}{n} \sum_{i=1}^n x_i^{\otimes d} - \frac{1}{m} \sum_{i=1}^m z_i^{\otimes d}$
  satisfies 
  $$
  \E \binom{p}{d}^{-1} \sum_{1 \le i_1 < \cdots < i_d \le p} E(i_1,\ldots,i_d)^2
  \lesssim 32^d \Big( \frac{\log \log n}{\kappa\log n} +
  \frac{p}{n^{1-\kappa}\e} + \frac{1}{m} \Big)
  $$   
  for all $d \le p/2$.
  The algorithm runs in time polynomial in $p$, $n$ and linear in $m$, and is independent of $d$.
\end{theorem}

A pseudocode description is given in Algorithm~\ref{algorithm2} below. 
\begin{remark}\label{lowerboundaverageerror}
When $n\geq(p/\epsilon)^{2}$, we can take $\kappa=1/3$ and $m=n$ in Theorem \ref{thm: Boolean private} and we have the following accuracy
\begin{equation}\label{loglogoverlogdecay}
  \E \binom{p}{d}^{-1} \sum_{1 \le i_1 < \cdots < i_d \le p} E(i_1,\ldots,i_d)^2
  \lesssim 32^d \frac{\log \log n}{\log n},
\end{equation} 
for all $d\leq p/2$. Can this decay in $n$ on the right hand side be improved if we use other polynomial time differentially private algorithms? Even for $d=1,2$, one cannot replace the right hand side by $n^{-a}$, for any $a>0$, with any polynomial time differentially private algorithm, assuming the existence of one-way functions. This is because if we could achieve this, then
\[\frac{1}{p^{2}}\mathbb{E}\max_{1\leq i_1,i_2\leq p}E(i_1,i_2)\leq\frac{1}{p^{2}}\E \sum_{i_1=1}^{p}\sum_{i_2=1}^{p} E(i_1,i_2)^2\lesssim\frac{1}{n^{a}},\]
so $\mathbb{E}\max_{1\leq i_1,i_2\leq p}E(i_1,i_2)\lesssim p^{2}/n^{a}$. But when $n\gg p^{2/a}$, this is impossible to achieve using any polynomial time differentially private algorithm if we assume the existence of one-way functions \cite{ullman2011pcps}. It remains an open question: among all polynomial time differentially private algorithms, what is the optimal decay in $n$ on the right hand side of (\ref{loglogoverlogdecay})?
\end{remark}
\begin{algorithm}[h!]
\caption{Differentially private Boolean synthetic data via microaggregation}
\label{algorithm2}
\begin{algorithmic}

\State {\bf Input:}  a  sequence of points $x_1,\ldots,x_n$ in the cube $\{0,1\}^p$ (true data);  $\epsilon,\kappa\in(0,1)$ (privacy); $m\in\mathbb{N}$ (number of points in the synthetic data).

\State {\bf Damped microaggregation}
\begin{enumerate}
\State\indent Redefine $x_{i}=\frac{1}{\sqrt{p}}x_{i}$ for $i=1,\ldots,n$.
\State\indent Compute the second-moment matrix $S=\frac{1}{n}\sum_{i=1}^{n}x_{i}x_{i}^{T}$.
\State\indent Let $t \coloneqq \left\lfloor \frac{\kappa\log n}{\log(7/\alpha)} \right\rfloor$ and $\alpha \coloneqq \frac{1}{(\log n)^{1/4}}$.
\State\indent Compute the second-moment matrix $S=\frac{1}{n}\sum_{i=1}^{n}x_{i}x_{i}^{T}$.
\State\indent Generate an orthogonal projection $P$ on $\mathbb{R}^{p}$ from $\mathrm{PROJ}(\frac{n\epsilon}{6t}S,t)$.
\State\indent Choose an $\alpha$-covering $\nu_1,\ldots,\nu_s \in \R^p$ of the unit Euclidean ball of the subspace $\ran(P)$. This is done by enumerating $B_2^t\cap(\alpha/\sqrt{t}) \Z^t$ and mapping it into $\ran(P)$ using any linear isometry.
\State\indent Construct the nearest-point partition $[n] = F_1 \cup \cdots \cup F_s$ for $Px_{1},\ldots,Px_{n}$ with respect to $\nu_1,\ldots,\nu_s$ as follows.
  For each $\ell \in [n]$, choose a point $\nu_j$ nearest to $x_\ell$ in the $\ell^2$ metric 
  and put $\ell$ into $F_j$. Break any ties arbitrarily. 
\State\indent Let $b=\sqrt{\frac{pn^{1-\kappa}}{\epsilon}}$.
\State\indent Perform damped microaggregation: compute $w_{j}=\frac{| F_{j}|}{n}$ and $\tilde{y}_{j}=\frac{1}{\max(| F_{j}|,b)}\sum_{i\in F_{j}}x_{i}\in\mathbb{R}^{p}, \quad j=1,\ldots,s$. 
\State\indent Generate an independent noise vector $\rho\in\mathbb{R}^{k}$ with $\rho_i \sim \Lap \Big( \frac{6}{n\e} \Big)$ for $i=1,\ldots,s$.
\State\indent For each $j=1,\ldots,s$, generate noise vectors $r_{j}\in\mathbb{R}^{p}$ with $r_{ji} \sim \Lap \Big( \frac{12\sqrt{p}}{b\e} \Big)$ for $i=1,\ldots,p$.
\State\indent Consider the simplex $\Delta = \Big\{ a=(a_1,\ldots,a_s) :\; \sum_{i=1}^s a_i = 1; \; a_i \ge 0 \; \forall i \Big\}$ and cube $K=\frac{1}{\sqrt{p}}[0,1]^p$.
\State\indent Compute the metric projections $\overline{w}=\pi_{\Delta,1}(w+\rho)$ and $\overline{y}_j = \pi_{K,2}(\tilde{y}_{j}+r_{j}), \quad j=1,\ldots,s$, by solving the convex minimizations in (\ref{piprojection})
\end{enumerate}

\State {\bf Bootstrapping} creates new data $u_1,\ldots,u_m$ by sampling from the points $\overline{y}_1,\ldots,\overline{y}_s$ with weights $\overline{w}_{1},\ldots,\overline{w}_s$, respectively.
\State {\bf Randomized rounding} maps the data $\{\sqrt{p} u_\ell\}_{\ell=1}^m \in [0,1]^p$ to data $\{z_j\}_{j=1}^m \in \{0,1\}^p$.

\State {\bf Output:} a sequence of points $z_1,\ldots,z_m$  in the cube $\{0,1\}^p$ (synthetic data) that satisfy the properties outlined in 
Theorem~\ref{thm: Boolean private}.
\end{algorithmic}

\end{algorithm}

\section*{Acknowledgement}

The authors would like to thank the referees and Hao Xing for their detailed and constructive feedback, which has led to definite improvements of some aspects of this paper.
M.B. acknowledges support from NSF DMS-2140592. T.S. acknowledges support from NSF-DMS-1737943 and NSF DMS-2027248. R.V. acknowledges support from NSF DMS-1954233, NSF DMS-2027299, U.S. Army 76649-CS, and NSF+Simons Research Collaborations on the Mathematical and Scientific Foundations of Deep Learning.





\end{document}